\documentclass[10pt, journal]{IEEEtran}

\newif\ifreport\reporttrue
\usepackage{graphicx}
\usepackage[center]{caption}
\usepackage{epsfig,latexsym,amsmath,amsfonts}
\usepackage{mathtools}

\usepackage{amssymb}
\usepackage{placeins}
\usepackage{subfigure}
\usepackage{verbatim}
\usepackage{cite,url}
\usepackage{epsf,bm}
\usepackage{color}
\usepackage{epstopdf}
\usepackage{romannum,hyperref}
\usepackage{dsfont}
\usepackage{stfloats}
\usepackage[english]{babel}
\usepackage{amsthm}
\newtheorem{definition}{Definition}

\newtheorem{assumption}{Assumption}
\usepackage[utf8]{inputenc}
\usepackage[english]{babel}
\usepackage{pifont}

\newtheorem{theorem}{Theorem}
\newtheorem{lemma}{Lemma}
\newtheorem{proposition}{Proposition}
\newtheorem{corollary}{Corollary}

\usepackage[lined, boxed, linesnumbered, ruled]{algorithm2e}
\IEEEoverridecommandlockouts
\begin{document}
\title{Optimal Sampling for Data Freshness: Unreliable Transmissions with Random Two-way Delay}
\author{Jiayu Pan, Ahmed M. Bedewy, 
        Yin Sun, \IEEEmembership{Senior Member, IEEE,}
        and Ness B. Shroff, \IEEEmembership{Fellow, IEEE}
\thanks{This paper was presented in part at IEEE INFOCOM 2022 \cite{pan2022optimizing}.

This work has been supported in part by NSF grants:  2112471, CNS- 2106932, CNS- 2106933, CNS-1955535, CNS-1901057, and CCF-1813050, and a grant from the Army Research Office: W911NF-21-1-0244.

J. Pan is with the Department of ECE, The Ohio State University,
Columbus, OH 43210 USA (e-mail: pan.743@osu.edu).

A. M. Bedewy is with the Department of ECE, The Ohio State University,
Columbus, OH 43210 USA (e-mail: bedewy.2@osu.edu).

Y. Sun is with the Department of ECE, Auburn University, Auburn, AL
36849 USA (e-mail: yzs0078@auburn.edu).

N. B. Shroff is with the Department of ECE and the Department of
CSE, The Ohio State University, Columbus, OH 43210 USA (e-mail:
shroff.11@osu.edu).}}

\maketitle

\begin{abstract}

In this paper, we \textcolor{black}{aim to design an optimal sampler for a system} in which fresh samples of a signal (source) are sent through an unreliable channel to a remote estimator, and acknowledgments are sent back over a feedback channel. Both the forward and feedback channels could have random transmission times \textcolor{black}{due to time varying channel conditions}. Motivated by distributed sensing, the estimator can estimate the real-time value of the source signal by combining the signal samples received through the channel and \textcolor{black}{the} noisy signal observations collected from a local sensor. We prove that the estimation error is a non-decreasing function of the Age of Information (AoI) for \textcolor{black}{the} received signal samples and design an optimal sampling strategy that minimizes the long-term average estimation error subject to a sampling rate constraint. \textcolor{black}{The} sampling strategy is also optimal for minimizing the long-term average of general non-decreasing functions of the AoI. The optimal sampler design follows a randomized threshold strategy: If the last transmission was successful, the source waits until the expected estimation error upon delivery exceeds a threshold and then sends out a new sample. If the last transmission fails, the source immediately sends out a new sample without waiting. The threshold is the root of a fixed-point equation and can be solved with low complexity (e.g., by bisection search). The optimal sampling strategy holds for general transmission time distributions of the forward and feedback channels. Numerical simulations are provided to compare different sampling policies.

\end{abstract}


\section{Introduction}\label{introduction}

Timely updates are crucial in many applications such as vehicular networks, wireless sensor networks, and UAV navigations. To achieve timely updates, we require the destination to receive fresh information from the remote source as quickly as possible. The information freshness is measured by age of information, or simply age, which has been widely explored in recent years (e.g., \cite{bedewy2019agee,bedewy2019minimizing,sun2018age,bedewy2019age,bedewy2021low,pan2021minimizing,zouminimizing,qian2020minimizing,yates2021age,sun2017remote,ornee2019sampling,tsai2020unifying,huang2020real,yates2015lazy,sun2017update,sun2019sampling,tsai2020age,tsai2021jointly,arafa2020timely,klugel2019aoi,pan2021minimizings,kaul2012real}). Age of information with the function of current time $t$ is defined as $\Delta_t = t-U_t$, where $U_t$ is the generation time of the freshest information data. 
In several different \textcolor{black}{queueing} systems, the Last-Generated, First-Served (LGFS) policy is shown to achieve age-optimality \cite{bedewy2019agee,bedewy2019minimizing,sun2018age}. Scheduling policies in various wireless networks are studied to minimize age  \cite{bedewy2019age,bedewy2021low,pan2021minimizing,zouminimizing,qian2020minimizing}.   
A literature review of \textcolor{black}{recent works} in age of information is provided in \cite{yates2021age}.

In \cite{sun2017remote} and \cite{ornee2019sampling}, a connection between age of information and remote estimation of time-varying processes (e.g., Wiener process or Ornstein-Uhlenbeck (OU) process) was established. One of the remote estimation objectives in these early studies \textcolor{black}{was} to design an optimal sampling policy to minimize the long-term average minimum mean square error (MMSE).
The MMSE is a function of \textcolor{black}{the} age if the sampling policy is independent of the signal being sampled \cite{tsai2020unifying,sun2017remote,ornee2019sampling,huang2020real}. 
Among these studies, the estimator obtains the exact signal samples subject to delay. However, the estimator neglects the instant and inexact signal samples.  
For example, in vehicular networks, the estimator can estimate a signal via both the exact signal samples from the remote sensor and the instant camera streaming from the close vehicle sensor over time.  
To consider both the delayed and instant signal samples, we will apply the Kalman Filter \textcolor{black}{\cite[Chapter 7]{poor2013introduction}} and study the relationship between the MMSE and age of information. 

The desire for timely updates and the study of the new remote estimation problem \textcolor{black}{necessitates considering general non-linear age functions in the development of optimal sampling policies.}
To reduce the age, we may require the source to wait before submitting a new sample \cite{yates2015lazy}. The study in \cite{sun2017update} generalized the result in \cite{yates2015lazy}, proposed an optimal sampling policy under a Markov channel with sampling rate constraint, and observed that the zero-wait policy is far from optimal if, for example, the transmission times are heavy-tail distributed or positively correlated.     
In \cite{sun2019sampling}, the authors provided a survey of the age penalty functions related to autocorrelation, remote estimation, and mutual information. The optimal sampling solution is a deterministic or randomized threshold policy based on the objective value and the sampling rate constraint. 
However, in real-time network systems, both the forward direction and the feedback direction have a random delay. Such a random two-way delay model was considered in e.g., \cite{tsai2020age,tsai2021jointly}.        
In \cite{tsai2020age}, the paper proposed a low complexity algorithm with a quadratic convergence rate to compute the optimal threshold.  
In \cite{tsai2021jointly}, an optimal joint cost-and-AoI minimization solution was provided for multiple coexisting source-destination pairs with heterogeneous AoI penalty functions. 
Although the above studies have \textcolor{black}{developed optimal} sampling strategies, they assume that the transmission process is reliable. However, due to the channel fading, the channel conditions are time-varying, and thus the transmission process is unreliable. 

\textcolor{black}{Recent studies \cite{arafa2020timely,klugel2019aoi} investigate sampling strategies while considering} unreliable transmissions. 
In \cite{arafa2020timely}, the authors considered quantization errors, noisy channel, and non-zero receiver processing time, and they established the relationship between the MMSE and age. For general age functions, they provided optimal sampling policies, given that the sampler needs to wait before receiving feedback. When the sampler does not need to wait, they \textcolor{black}{provided} enhanced sampling policies that perform better than previous ones.  
In \cite{klugel2019aoi}, the authors chose \emph{idle} or \emph{transmit} at each time slot to minimize joint age penalty and transmission cost. The optimality of a threshold-based policy is shown, and the policy's threshold is computed efficiently. 
Nevertheless, in practice, transmission delays are random rather than constant because of \textcolor{black}{congestion}, random sample sizes, etc, \textcolor{black}{which is a critical challenge facing the design of sampling strategies.}

\textcolor{black}{To address the aforementioned challenges, we investigate how to design optimal sampling strategies} in wireless networks under the following more realistic (and general) conditions that have largely been unexplored: unreliable transmissions and random delay in both forward and feedback directions. 
Early studies on optimizing sampling assuming reliable channels with random delays 
have shown that the sampling problem is decomposed into a per-sample problem. The per-sample problem can be further solved by optimization theory (e.g., \cite{yates2015lazy,sun2017update,sun2019sampling,tsai2020age}) or optimal stopping rules (e.g., \cite{tsai2020unifying,sun2017remote,ornee2019sampling}). 
Similarly, our problem assuming an unreliable channel is equivalent to a per-epoch problem containing multiple samples until successful packet delivery.
Therefore, the per-epoch problem is a Markov Decision Process (MDP) with an uncountable state space, which is \textcolor{black}{a} key difference with past works, (e.g., \cite{yates2015lazy,sun2017update,sun2019sampling,tsai2020age,tsai2020unifying,sun2017remote,ornee2019sampling}) and faces the curse of dimensionality.\footnote{\textcolor{black}{We further compare our technical differences with past works in Section \ref{discussion}}.}
The main contributions of this paper are stated as follows:        


\begin{itemize}

\item We first formulate the problem where the estimator estimates a signal in real-time by combining noisy signal observations from a local sensor and accurate signal samples received from a remote sensor. 
We show that if the sampling policy is made independently of the signal being sampled, the MMSE equals an increasing function of the age of \textcolor{black}{the} received signal samples.  

\item For general nonlinear age functions, or simply age penalty functions, we provide an exact solution for minimizing these data freshness metrics.  
The optimal sampling policy has a simple threshold-type structure, and the threshold can be efficiently computed by bisection search and fixed-point iterations. We uncover the following interesting property: if the last transmission is successful, the optimal policy may wait for a positive time period before generating the next sample and sending it out; otherwise, no waiting time should be added. The key technical approach developed in our results is given as follows: (i) The value function of the proposed policy is an exact solution to the Bellman equation. (ii) Under the contraction mapping assumption, the solution to the Bellman equation is unique, which guarantees optimality of our proposed threshold-based policy. 
Our results hold for (i) general non-decreasing age penalty functions, (ii) general delay distributions of both the forward and feedback channels, (iii) sampling problems both with or without a sampling rate constraint. Therefore, our paper extends previous studies on sampling for optimizing age (e.g., \cite{yates2015lazy,sun2017update,sun2019sampling,tsai2020age,arafa2020timely,klugel2019aoi}). Although our sampling problem is \textcolor{black}{in} continuous time, it can be easily reduced to be in discrete time.

\item When there is no sampling rate constraint, we provide \textcolor{black}{necessary and sufficient conditions on the optimality of the zero-wait sampling policy \cite{yates2021age}} based on the choice of age penalty function, forward and feedback channels. 
Finally, numerical simulations show that our optimal policy can reduce the age compared with other approaches. 

\end{itemize}



\section{Estimation and the AoI}\label{estimation}

\begin{figure}[t]
\includegraphics[scale=.45]{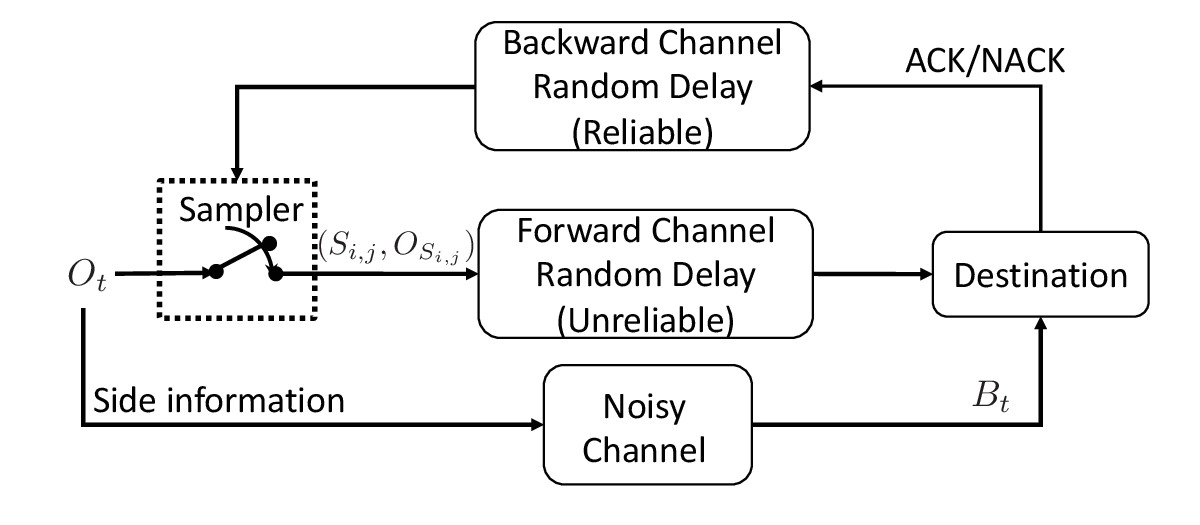}
\centering
\captionsetup{justification=justified}
\caption{System model.}
 \label{model}
\end{figure}

\subsection{System Model}\label{system-a}

Consider a status update system that is composed of a source, a destination, a source-to-destination channel, and a destination-to-source channel, as is illustrated in Fig. \ref{model}. The source process $O_t$ is sampled and delivered to the destination via the forward channel. 
The forward channel suffers from i.i.d. transmission failures, where $\alpha\in[0,1)$ is the probability of failure.
Upon each delivery, the destination then sends an $1$-bit feedback message denoting whether the transmission is successful (ACK) or unsuccessful (NACK). The feedback is sent via the feedback channel that is reliable with an i.i.d. random delay. 

To clarify the system model, we set $i\in \{1,2,\ldots\}$ as the label of a successful delivery in chronological order. Let us denote the $i$th \emph{epoch} to be the time period between the $(i-1)$th and the $i$th successful deliveries. We denote $M_i$ as the total number of samples attempted during the $i$th epoch. 
Then, the $M_i$'s are i.i.d. and has a geometric distribution with parameter $1-\alpha$. 
We use $j$ to describe the indices of samples at the $i$th epoch, where we have $1\leq j \leq M_i$. The case $j=1$ implies that the previous sample is successfully transmitted to the destimation.  
Upon delivery, the destination immediately sends the feedback to the sampler and arrives at time $A_{i,j}$ via the backward channel with an i.i.d. delay $X_{i,j}$, which satisfies $\mathbb{E} [X_{i,j}]<\infty$. 
Then, the $j$th sample in the $i$th epoch is generated at $S_{i,j}$ and is delivered at $D_{i,j}$ through the forward channel with an i.i.d. delay $Y_{i,j}$, which satisfies $\mathbb{E} [Y_{i,j}]<\infty$. 

We assume that the backward delays $X_{i,j}$'s and forward delays $Y_{i,j}$'s are mutually independent.
In addition, the source generates a sample after receiving the feedback of the previous sample\footnote{This assumption arises from the stop-and-wait mechanism. When the backward delay $X_{i,j}=0$, the policy that samples ahead of receiving feedback is always suboptimal. The reason is that such a policy takes a new sample when the channel is busy and can be replaced by another policy that samples at the exact time of receiving feedback \cite{sun2019sampling}. When $X_{i,j}\ne 0$, however, it may be optimal to transmit before receiving feedback, which is out of the scope of this paper.}, i.e., $S_{i,j}\ge A_{i,j}$. 
 In other words, we have a non-negative waiting time $Z_{i,j}$ for all epoch $i$ and sample $j$. 
Thus, the forward channel is always available for transmission at $S_{i,j}$, and the delivery time $D_{i,j}$ satisfies $D_{i,j} = S_{i,j}+Y_{i,j}$.  
By Wald's equation, the total transmission delay needed in each epoch has a finite expectation: 
\begin{equation}\label{walds}
  \mathbb{E} \left[   \sum_{j=1}^{M_i} \left( X_{i,j}+Y_{i,j} \right) \right] =  \mathbb{E} \left[ X_{i,j}+Y_{i,j}  \right] \mathbb{E} \left[ M_i  \right]<\infty.
\end{equation}   

\begin{figure}[t]
\includegraphics[scale=.44]{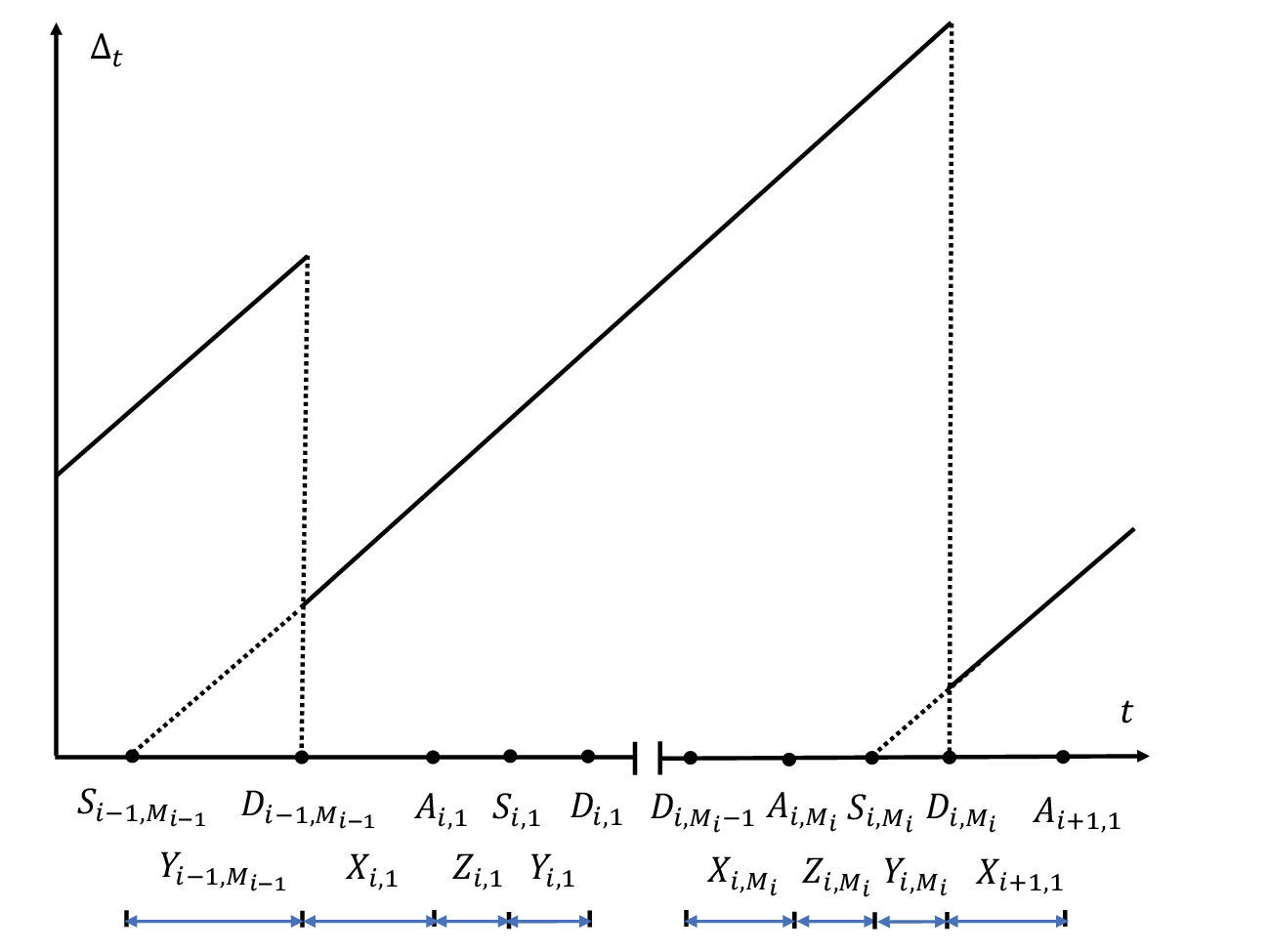}
\centering
\captionsetup{justification=justified}
\caption{Evolution of the age $\Delta_t$ over time. The $i$th epoch starts from $D_{i-1,M_{i-1}}$ to $D_{i,M_i}$.}
 \label{fig-age}
\end{figure}

Age of information (or simply \emph{age}) is the metric for evaluating data freshness and is equal to the time elapsed between the current time $t$ and the generation time of the freshest delivered packet \cite{kaul2012real}. Let $U_t = \max_i \{ S_{i,M_i}: D_{i,M_i}\le t \}$. Note that only the $M_i$th sample is successfully delivered for the $i$th epoch. Then, the age of information $\Delta_t$ at the current time $t$ is defined as
 \begin{equation}\label{age}
 \Delta_t = t - U_t. 
 \end{equation} 
 We plot the evolution of the age \eqref{age} in Fig. \ref{fig-age}. 
 Upon each successful delivery time $D_{i,M_i}$, the age decreases to $Y_{i,M_i}$, the transmission delay of the newly generated packet. At other time, the age increases linearly over time. The age is updated at the beginning of each epoch and keeps increasing during the epoch. Hence, the age is also determined by 
  \begin{equation}
 \Delta_t = t - S_{i,M_i}, \text{ if }  D_{i,M_i}\leq t<D_{i+1,M_{i+1}}.
 \end{equation}

\subsection{Remote Estimation and Kalman Filter}\label{remoteestimation}

We first introduce some notations. For any multi-dimensional vector $O$, we denote $O^T$ as the transpose of $O$. We denote $\bm{I}_{n\times n}, \bm{0}_{n\times m}$ as the $n\times n$ identity matrix and $n\times m$ zero matrix, respectively. For a given $n\times n$ matrix $\bm{N}$, we set $tr(\bm{N})$ as the trace of $\bm{N}$, i.e., the summation of the diagonal elements of $\bm{N}$. 

In this subsection, the source process $O_t$ is an $n$-dimensional diffusion process that is defined as the solution to the following stochastic differential equation:  

\begin{equation}
dO_t=- \bm{\Theta} O_t dt+ \bm{\Sigma} d W_t,\label{ou-process}
\end{equation}
where $\bm{\Theta}$ and $\bm{\Sigma}$ are $n\times n$ matrices, and $W_t$ is the $n$-dimensional Wiener process such that $\mathbb{E}[W_tW_s^T]= \bm{I}_{n\times n} \min \{s,t \}$ for all $0\leq t,s\leq\infty$. The process $O_t$ represents the behavior of many physical systems \textcolor{black}{such as the motion of a Brownian particle under friction and the motion of the monomers in dilute solutions\cite{ottinger2012stochastic}.}  
At the destination, there is an estimator that provides estimations according to the received samples. One key difference from previous works (e.g., \cite{tsai2020unifying,sun2017remote,ornee2019sampling,klugel2019aoi}) is that the estimator not only receives the accurate samples $O_{S_{i,j}}$ at time $S_{i,j}$ but also has an instant noisy observation $B_t$ of the process $O_t$, as is illustrated in Fig. \ref{model}. 
The observation process $B_t$ is an $m$-dimensional vector, modeled as 
\begin{equation}
B_t=\bm{H}O_t+V_t,
\end{equation}
where $\bm{H}$ is an $n\times m$ matrix and $V_t$ is a zero mean white noise process such that for all $t,s \ge 0$, 
\begin{equation}
\mathbb{E}[V_tV_s^T]= \left\{
 \begin{array}{lll} 
 \bm{R} & t=s; \\ \bm{0}_{m\times m} & t\ne s, 
 \end{array} \right.
\end{equation} $\bm{R}$ is an $m\times m$ positive definite matrix.
We suppose that $W_t$ and $V_t$ are uncorrelated such that for all $ t,s \ge 0$, $\mathbb{E}[W_tV_s^T]=\bm{0}_{n\times m}$.

The estimator provides an estimate $\hat{O}_t$ for the minimum mean squared error (MMSE) $\mathbb{E}[ || O_t-\hat{O}_t ||^2]$ based on the causally received information.
Compared to \cite{ornee2019sampling}, the MMSE in our study can be reduced due to the additional observation process $B_t$.
Using the strong Markov property of $O_t$ \cite[Eq. (4.3.27)]{peskir2006optimal} and the assumption that the sampling times are independent of $O_t$, as is shown
\ifreport
in Appendix \ref{hat_otapp},
\fi
the MMSE estimator is determined by  
\begin{equation}
\hat{O}_t = \mathbb{E} \left[ O_t | \{ B_\tau \}_{S_{i,M_i}\le \tau\le t}, O_{S_{i,M_i}} \right], t\in  [D_{i,M_i},D_{i+1,M_{i+1}}). \label{hat_ot}
\end{equation}
By \eqref{hat_ot}, we find that $\hat{O}_t$ is equal to \textcolor{black}{the estimate produced by} the \emph{Kalman filter}\textcolor{black}{\cite[Chapter 7]{poor2013introduction}.}
Therefore, in this work, we use the Kalman filter as the estimator. At time $t$, the Kalman filter utilizes both the exact sample $O_{S_{i,M_i}}$ and noisy observation $B_t$ and provides the minimum mean squared error (MMSE) estimation $\hat{O}_t$. Let $\bm{N}_t\triangleq\mathbb{E}[ (O_t-\hat{O}_t) (O_t-\hat{O}_t)^T]$ be the covariance matrix of the estimation error $O_t-\hat{O}_t$. Hence, \textcolor{black}{$\mathbb{E}[ || O_t-\hat{O}_t ||^2]=tr(\bm{N}_t)$}. 

According to \eqref{hat_ot}, the estimation process works as follows: Once a sample is delivered to the Kalman filter at time $D_{i,M_i}$, the Kalman filter re-initiates itself with the initial condition $\bm{N}_t=\bm{0}_{n\times n}$ when $t=S_{i,M_i}$ and starts a new estimation session. Then, during the time period $ [D_{i,M_i},D_{i+1,M_{i+1}})$, the Kalman filter uses the causal observations $\{ B_{\tau}: S_{i,M_i} \le \tau \le t \}$ to estimate the process $O_t$.

\begin{proposition}\label{lim1}
The MMSE $tr(\bm{N}_t)$ of the process $O_t$ is a non-decreasing function of the age $\Delta_t$.
\end{proposition}
\begin{proof}
\ifreport
See Appendix \ref{lim1app}. 
\fi
\end{proof}
As a result of Proposition \ref{lim1}, when the sampling times $S_{i,j}$'s are independent of $O_t$, the MMSE is still a non-decreasing function of the age $\Delta_t$. When $S_{i,j}$'s are correlated to $O_t$, the MMSE is not necessary a function of $\Delta_t$.

In the one-dimensional case, where $n=m=1$, we use scalars $\theta, \sigma, h, r, n_t$ to replace the matrices $\bm{\Theta}, \bm{\Sigma}, \bm{H}, \bm{R}, \bm{N}_t$, respectively. The Ornstein–Uhlenbeck (OU) process is defined as a one-dimensional special case of diffusion process \eqref{ou-process} where $\theta>0$\cite{finch2004ornstein}. Then, we have
\begin{proposition}\label{lemma-lim2} 
Suppose that $n=m=1$ and $\theta>0$. Then, for $t\in[D_{i, M_i},D_{i+1,M_{i+1}})$ and $ i=0,1,2,\ldots$, the MMSE $n_t$ of the OU process $O_t$ is given by
\begin{align}
n_t=\bar{n}-\frac{1}{l+\left(\frac{1}{\bar{n}}-l\right)e^{2\sqrt{\theta^2+\frac{\sigma^2 h^2}{r}}\Delta_t}}, \label{Lem1eq}
\end{align}
where $\Delta_t=t-S_{i,M_i}$, 
\begin{align}
\bar{n} & =\frac{-\theta r+\sqrt{(\theta r)^2+\sigma^2 rh^2}}{h^2},\\
l & = \frac{h^2}{2\sqrt{(\theta r)^2+\sigma^2 r h^2}}. \label{remote_9}
\end{align} 
Moreover, $n_t$ in \eqref{Lem1eq} is a bounded and non-decreasing function of the age $\Delta_t$.
\end{proposition}
\begin{proof}
\ifreport
See Appendix \ref{lemma-lim2app}.
\fi
\end{proof}

When the side observation has zero knowledge of $O_t$, i.e., $h=0$ for $t\ge 0$, then the estimator $\hat{O_t}$ is equal to that in \cite{ornee2019sampling}.
Therefore, Proposition \ref{lemma-lim2} reduces to \cite[Lemma 4]{ornee2019sampling}, i.e.,
the MMSE $n_t$ is given by
\begin{align}
n_t=\frac{\sigma^2}{2\theta}\left(1-e^{-2\theta\Delta_t}\right), \label{Lim2_eq}
\end{align}
moreover, $n_t$ for $h=0$ is a bounded and non-decreasing function of age $\Delta_t$.


\section{Problem Formulation for General Age Penalty}\label{section-system}

The function in Proposition \ref{lemma-lim2} is not the only choice of nonlinear age functions.
In this paper, to achieve data freshness in various applications, we consider a general type of age penalty function. 
The age penalty function $p:[0,\infty)\rightarrow \mathbb{R}$ is assumed to be non-decreasing and need not be continuous or convex.
We further assume that 
$\mathbb{E} \Big[  \int^{\delta+\sum_{j=1}^{M_i}(X_{i,j}+Y_{i,j})}_{\delta}  p(t) dt \Big] < \infty$ and $\mathbb{E} \Big[  p \left( \delta+\sum_{j=1}^{M_i}(X_{i,j}+Y_{i,j}) \right) dt \Big] < \infty$ for any given $\delta$.  \vspace{0.06\baselineskip}

We list another two categories of applications for the age penalty functions. \textcolor{black}{First, the age penalty functions can be linear, polynomial, or exponential, depending on the dissatisfactions of the stale information updates in multiple practical settings such as the Internet of Things \cite{park2020centralized}.} 
Second, some applications are shown to be closely related to nonlinear age functions, such as auto-correlation function of the source, remote estimation, and information based data freshness metric \cite{sun2019sampling}.




We then define the sampling policies below. We denote $\mathcal{H}_{i,j}$ as the sample path of the history information previous to $A_{i,j}$, including sampling times, forward channel conditions, and channels delays. 
We denote $\Pi$ as the collection of sampling policies $\{ S_{i,j} \}_{i,j}$ such that $S_{i,j} \ge A_{i,j}$ for each $(i,j)$, and $S_{i,j}(ds_{i,j} |\mathcal{H}_{i,j})$ is a \emph{Borel measurable stochastic kernel} \cite[Chapter 7]{bertsekas2004stochastic} for any possible $\mathcal{H}_{i,j}$.
Further, we assume that $T_i= S_{i,M_i}-S_{i-1,M_{i-1}}$ is a regenerative process: there exists an increasing sequence $0\le k_1 < k_2< \ldots$ of finite random variables such that the post-$k_j$ process $\{ T_{k_j+i},i=0,1,\ldots \} $ has the same distribution as the post-$k_1$ process  $\{ T_{k_1+i},i=0,1,\ldots \}$ and is independent of the pre-$k_j$ process $ \{ T_i, i=1,2, \ldots, k_j-1 \}$; in addition, $\mathbb{E} \left[  k_{j+1}-k_j \right]<\infty$, $\mathbb{E} \left[  S_{k_1,M_{k_1}} \right]<\infty$ and $0<\mathbb{E} \left[  S_{k_{j+1},M_{k_{j+1}}}-S_{k_j,M_{k_j}} \right]<\infty$, $j=1,2, \ldots$\footnote{ In this paper, we will optimize $\limsup_{T\rightarrow \infty} (1/T) \mathbb{E} \left[  \int_{0}^{T} p(\Delta_t) dt \right]$. However, a nicer objective is to optimize $ \lim_{n\rightarrow \infty}  \mathbb{E} \left[  \int^{D_{n,M_n}}_{0}  p(\Delta_t) dt \right] $ $/ \mathbb{E} \left[  D_{n,M_n}  \right] $. If $T_i$ is a regenerative process, then the two objective functions are equal \cite{haas2006stochastic}, \cite{gallager2013stochastic}. If no conditions are applied, they are different.}


The authors in \cite{tsai2020unifying} have stated that: to reduce the estimation error related to the Wiener process, it may be optimal to wait on both the source and the destination before transmission. 
However, in this paper, it is sufficient to only wait at the source to minimize the age. To validate this statement, consider any policy that waits on both the source and the destination. We first remove the waiting time at the destination. Then, at the source, we add up the removed waiting time. 
The replaced policy we propose has the same age performance as the former one.

Our objective in this paper is to optimize the long-term average expected age penalty under a sampling rate constraint:
\begin{align}
p_{\text{opt}} = & \inf_{\pi \in \Pi} \limsup_{T\rightarrow \infty}  \frac{1}{T} \mathbb{E} \left[  \int_{0}^{T} p(\Delta_t) dt \right], \label{avg} \\
& \  \text{s.t.} \   \limsup_{T\rightarrow \infty} \frac{1}{T} \mathbb{E} \left[ C(T)\right] \le f_{\text{max}}. \label{avg-constraint}
\end{align}
Here, $C(T)$ is the total number of samples taken by time $T$, and $f_{\text{max}}$ is the maximum allowed sampling rate. The constraint \eqref{avg-constraint} is added because in practice, the sensor may need to keep working for a long time with limited amount of energy.  
To avoid triviality, the optimal objective value $p_{\text{opt}}$ in \eqref{avg} satisfies $p_{\text{opt}}<\bar{p}$, where $\bar{p}=\lim_{\delta\rightarrow \infty} p(\delta)$.



\subsection{An Additional Assumption and Its Rationale}\label{additional-assumption}


We will utilize the following assumption in this paper.
\begin{assumption}\label{ass1}
If $\alpha>0$, the backward delay $X_{i,j}\in [0, \bar{x}]$, and the waiting time $Z_{i,j}\in [0,\bar{z}]$ for all $i,j$. For any positive $\bar{x}, \bar{z}$ (that can be sufficiently large),
there exists an increasing positive function $v(\delta)$ such that the function $G(\delta) =  \mathbb{E} \left[   \int_{\delta}^{\delta+\bar{x}+\bar{z}+Y_{i,j}} | p(t) | dt \right] $ satisfies $\max_{\delta\ge0} | G(\delta)/v(\delta) | <\infty$. 
In addition,
there exists $\rho \in (0,1)$ and a positive integer $m$, such that  
\begin{equation}
\alpha^m \frac{ \mathbb{E} \left[  v(\delta+m \bar{x}+m\bar{z}+\sum_{j=1}^{m}Y_{j}) \right]}{v(\delta)} \le \rho
\end{equation}
holds for all $\delta\ge 0$, where $Y_1, \ldots, Y_m$ are an i.i.d. sequence with the same distribution as the $Y_{i,j}$'s.  
\end{assumption}
When the forward channel is reliable, i.e., $\alpha=0$, then Assumption \ref{ass1} is negligible by letting $v(\delta) = G(\delta)$. Thus, Assumption \ref{ass1} restricts on the choices of age penalty $p(\cdot)$ when $\alpha>0$. Note that the optimal sampling policy of the cases $\alpha=0$ and $X_{i,j}=0$ has been solved in \cite{sun2017update,sun2019sampling}.

In the following corollary, we provide a list of age penalties $p(\cdot)$ that Assumption \ref{ass1} is satisfied for $\alpha>0$.     
\begin{corollary}\label{cor-ass}
For any one of the following conditions, Assumption \ref{ass1} holds: 

(a) The penalty function $p(\cdot)$ is bounded, i.e., $\bar{p}<\infty$.

(b)
There exists $n > 0$ such that $p(\delta) = O (\delta^n)$,\footnote{ We denote $f(\delta) = O (g(\delta))$ if there exists some nonnegative constants $c$ and $\delta'$ such that $ |f(\delta)| \le c |g(\delta)| $ for all $\delta>\delta'$.}
 and the $Y_{i,j}$'s have a finite $n+1$-moment, i.e., $ \mathbb{E} \left[  Y_{i,j}^{n+1} \right]<\infty$. 


(c) There exists $a>0$ and $b<1$ such that $\int p(\delta) d\delta = O( e^{a\delta^b})$ and the $Y_{i,j}$'s are bounded. 
\end{corollary}
\ifreport
\begin{proof} 
See Appendix \ref{cor-assapp}.
\end{proof} 
\fi

Most of the literatures of MDP have shown that the \emph{value function} of an optimal policy is the solution to the \emph{Bellman equation}. In this paper, we figure out a policy and its value function that is indeed the solution to the Bellman equation. If the Bellman equation has a unique solution, then our proposed policy is optimal. Otherwise, we cannot guarantee the optimality of our proposed policy. Assumption~\ref{ass1} arises from the contraction mapping assumption \cite{bertsekas1995dynamic2,bertsekas2018abstract} that guarantees that the Bellman equation has a unique solution.
In other words, Assumption~\ref{ass1} is a sufficient condition for the Bellman equation to have a unique solution. Corollary~\ref{cor-ass} implies that there are a wide range of age penalty functions that satisfy Assumption~\ref{ass1}. For example, the age penalty function derived in Proposition~\ref{lemma-lim2} satisfies Assumption~\ref{ass1}. Indeed, Assumption~\ref{ass1} holds if the age penalty function grows exponentially at some bounded intervals. For all cases of the age penalty functions we have mentioned, the constants $\bar{z},\bar{x}$ can be \emph{sufficiently large}. Therefore, in this paper, we set the constants $\bar{z},\bar{x}$ to be sufficiently large.


\section{Optimal Sampling policy}\label{main}

In this section, we provide an optimal solution to \eqref{avg}. The optimal solution is described by the waiting times $Z_{i,j}'s$ throughout this paper.

\subsection{Optimal Sampling Policy without Sampling Rate Constraint}\label{no-sampling}
When there is no sampling rate constraint, i.e., $f_{\text{max}}=\infty$, we have the following result:
\begin{theorem}\label{theorem1}
If $f_{\text{max}}=\infty$, $p(\cdot)$ is non-decreasing, the $Y_{i,j}$'s are i.i.d. with finite mean $\mathbb{E}[Y_{i,j}]<\infty$, the $X_{i,j}$'s are i.i.d. with finite mean $\mathbb{E}[X_{i,j}]<\infty$, the $Y_{i,j}$'s and the $X_{i,j}$'s are mutually independent, and Assumption \ref{ass1} holds, then the optimal solution to \eqref{avg} is given by
\begin{align}
\nonumber  & Z_{i,1}(\beta) =  \inf_{z} \Big \{  z\ge 0:  \\ & 
 \mathbb{E}_{Y'} \left[ p( Y_{i-1,M_{i-1}}+X_{i,1} +z+Y')   \ \big{|} \  Y_{i-1,M_{i-1}}, X_{i,1} \right] \ge \beta \Big \}, \label{thm1-beta}  \\
& Z_{i,j}(\beta)  =  0 \ \ \ j = 2,3,\ldots, \label{thm1-beta2}
\end{align}
${Y'} = Y_{i,1}+ \sum_{j=2}^{M_i} (X_{i,j}+Y_{i,j})$,\footnote{In this paper, we set the summation operator $\sum_{j = a}^{b}$ to be $0$ if $b<a$ for any given integers $a,b$.} and $\beta$ is the unique solution to 
\begin{align}
\nonumber & \mathbb{E} \left[  \int^{Y_{i-1,M_{i-1}}+X_{i,1} +Z_{i,1}(\beta)+Y'}_{Y_{i-1,M_{i-1}}}  p(t) dt \right] \\  & - \beta  \mathbb{E} \left[ X_{i,1} +Z_{i,1}(\beta)+Y'  \right]=0.  \label{root}
\end{align}
Moreover, $\beta=p_{opt}$ is the optimal objective value of \eqref{avg}.
\end{theorem}
\begin{proof}
See Section \ref{section-proof}.
\end{proof}

In Theorem~\ref{theorem1}, the case $j=1$ in \eqref{thm1-beta} means that the previous transmission (of the $M_{i-1}$th sample in the $(i-1)$th epoch) is successful, and the system starts the new epoch from $i-1$ to $i$. Since the age drops to $Y_{i-1,M_{i-1}}$ at the successful delivery time $D_{i-1,M_{i-1}}$, the current age state at arrival time $A_{i,1}$ is $Y_{i-1,M_{i-1}}+X_{i,1}$.
The case $j=2,3,\ldots$ in \eqref{thm1-beta2} means that the previous transmission is unsuccessful, and the system stays within epoch $i$.

Theorem \ref{theorem1} provides an optimal policy with an interesting structure. First, by \eqref{thm1-beta}, in each epoch, the optimal waiting time for the first sample $Z_{i,1}(\beta)$ has a simple threshold type structure on the current age $Y_{i-1,M_{i-1}}+X_{i,1}$. Since the waiting times for $j=2,3,\ldots$ are zero, $Y'$ is the remaining transmission delay needed for the next successful delivery. Note that $\beta$ is equal to the optimal objective value $p_{\text{opt}}$ in problem \eqref{avg}. Therefore, the waiting time $Z_{i,1}(\beta)$ in \eqref{thm1-beta} is chosen such that the expected age penalty upon delivery is no smaller than $p_{\text{opt}}$.
Second, by \eqref{thm1-beta2}, the source sends the packet as soon as it receives negative feedback, i.e., the previous transmission is not successful. This is quite different from most of the previous works assuming reliable channels, e.g.,\cite{sun2017update,sun2019sampling,tsai2020age,tsai2021jointly}, where for all samples, the source may wait for some time before transmitting a new sample.  

We call a sampling policy to be \emph{stationary} if each sampling time is decided by the current age state and the previous backward delay. We call a sampling policy to be \emph{deterministic} if each sampling time chooses a value with probability $1$ (w.p. $1$).    
We remind that the optimal policy we proposed in Theorem~\ref{theorem1} is \emph{stationary and deterministic}. This stationary and deterministic policy depends only on the current age state and the previous backward delay, not on the sample index $j$. For example, when $j=1$, the previous backward delay is $X_{i,1}$, and the current age state is $Y_{i-1,M_{i-1}}+X_{i,1}$. For general value of $j$, the previous backward delay is $X_{i,j}$, and we suppose that the current age state is $\Delta_{i,j}+X_{i,j}$. Then, the stationary and deterministic policy, which has an equivalent form of \eqref{thm1-beta},\eqref{thm1-beta2} in Theorem \ref{theorem1}, is as follows:
\begin{align}
\nonumber  Z_{i,j}(\beta) = & \inf_{z} \Big \{  z\ge 0:  \\ & 
 \mathbb{E}_{Y'} \left[ p( \Delta_{i,j}+X_{i,j} +z+Y')  \ \big{|} \  \Delta_{i,j}, X_{i,j} \right] \ge \beta \Big \}.
\end{align}      


 \begin{algorithm}[!htbp]\label{algbise1}
\caption{Bisection method for solving \eqref{root}}
\textbf{Given} function $f(\beta)=f_1(\beta)- \beta f_2(\beta)$. $k_1$ close to $\underline{p}$, $k_2$ close to $\bar{p}$, $k_1<k_2$, and tolerance $\epsilon$ small. 

\textbf{repeat}

\qquad $\beta=\frac{1}{2}(k_1+k_2)$

\qquad \textbf{if} $f(\beta)<0$: $k_2=\beta$. \textbf{else} $k_1=\beta$ 

\textbf{until} $k_2-k_1<\epsilon$

\textbf{return} $\beta$

\end{algorithm}

The root of $\beta$ in \eqref{root} can be solved efficiently. 
According to \eqref{root}, we can use a low complexity algorithm such as bisection search and fixed-point iterations to obtain the optimal objective value $p_{\text{opt}}$. The bisection search approach to solving $p_{\text{opt}}$ is illustrated in Algorithm 1. 
For simplicity, we set 
\begin{align}
f_1(\beta) & = \mathbb{E} \left[  \int^{Y_{i-1,M_{i-1}}+X_{i,1} +Z_{i,1}(\beta)+Y'}_{Y_{i-1,M_{i-1}}}  p(t) dt \right],\\
f_2(\beta) & =   \mathbb{E} \left[ X_{i,1} +Z_{i,1}(\beta)+Y'  \right].
\end{align}
Then, the function $f(\beta) \triangleq f_1(\beta)-\beta f_2(\beta)$ satisfies the following mathematical property:
 \begin{lemma}\label{lemma-hbeta}
 (1) $f(\beta)$ is concave, and strictly decreasing in $\beta \in [\underline{p},\bar{p}) \cap \mathbb{R}$, where $\underline{p} = p(0)$ and $\bar{p} = \lim_{\delta\rightarrow \infty}p(\delta)$.

(2) There exists a unique root $\beta \in [\underline{p},\bar{p}) \cap \mathbb{R}$ such that $f(\beta)=0$.
 \end{lemma}  
 \begin{proof}
 \ifreport
See Appendix \ref{lemma-hbetaapp}.
\fi
 \end{proof}
\noindent Therefore, the solution to Algorithm 1 is unique.


One common sampling policy is the zero-wait policy, which samples the packet once it receives the feedback, i.e., $Z_{i,j} = 0$ for all $(i,j)$ \cite{yates2021age}.
The zero-wait policy maximizes the throughput and minimizes the delay. However, by Theorem \ref{theorem1}, the zero-wait policy may be suboptimal on age.  
The following result provides the necessary and sufficient condition when the zero-wait policy is optimal. 
\begin{corollary}\label{equiv-zerowait}
If $f_{\text{max}}=\infty$, $p(\cdot)$ is non-decreasing, the $Y_{i,j}$'s are i.i.d. with finite mean $\mathbb{E}[Y_{i,j}]<\infty$, the $X_{i,j}$'s are i.i.d. with finite mean $\mathbb{E}[X_{i,j}]<\infty$, the $Y_{i,j}$'s and the $X_{i,j}$'s are mutually independent, and Assumption \ref{ass1} holds, then the zero-wait policy is optimal if and only if
\begin{align}\label{cor2}
\text{ess} \inf \mathbb{E}_{Y'} \left[  p( Y+X + Y') \mid Y,X \right] \ge  \frac{ \mathbb{E} \left[  \int^{Y+X + Y'}_{Y}  p(t) dt \right]  }  { \mathbb{E} \left[ X+Y' \right] },
\end{align} 
where $Y' = Y_{i,1}+ \sum_{j=2}^{M_i} (X_{i,j}+Y_{i,j})$, $Y = Y_{i-1,M_{i-1}}, X=X_{i,1}$
and we denote ess$\inf E= \inf \left\{ e: \mathbb{P}(E \le e) >0 \right\}$ for any random variable $E$.
\end{corollary}
\begin{proof}
\ifreport
See Appendix \ref{equiv-zerowaitapp}.
\fi
\end{proof}


When the channel delays are constant, we can get from Corollary \ref{equiv-zerowait} that 

\begin{corollary}\label{deter}
If $f_{\text{max}}=\infty$, $p(\cdot)$ is non-decreasing and satisfies Assumption \ref{ass1}, and the $Y_{i,j}$'s, $X_{i,j}$'s are constants, then the zero-wait policy is the solution to problem \eqref{avg}.
\end{corollary}
\begin{proof}
\ifreport
See Appendix \ref{deterapp}.
\fi
\end{proof}

Theorem \ref{theorem1} is an extension to \cite{sun2019sampling,tsai2020age}.
When the forward channel is reliable, i.e., $M_i=1$ for all $i$ or $\alpha=0$, Theorem \ref{theorem1} can be reduced to the result in \cite{tsai2020age}. 
Further, we extend \cite{tsai2020age} in two folds: (i) The age penalty $p(\cdot)$ is allowed to be negative or discontinuous. (ii) The channel delays $Y_{i,1}, X_{i,1}$ have a finite expectation and do not need to be bounded. Note that when $M_i=1$, Assumption \ref{ass1} is negligible. 
When $M_i=1$, and there is no backward delay ($X_{i,1}=0$), our result reduces to \cite[Theorem 1]{sun2019sampling}. 

The study in \cite[Theorem 2]{arafa2020timely} proves the optimality of the zero-wait policy among the deterministic policies under an unreliable forward channel. This result corresponds to Corollary \ref{deter}, a special case of Theorem \ref{theorem1}. Our paper extends \cite{arafa2020timely} in two folds: (i) We allow the policy space $\Pi$ to be randomized. Among randomized policies, due to the disturbances on the previous sampling times, the current sampling time is dependent on the previous ones, which is different from \cite{arafa2020timely}. (ii) We consider random two-way delays, extending the constant one-way delay in \cite{arafa2020timely}. 


\subsection{Optimal Sampling Policy with Sampling Rate Constraint}\label{with-sampling}
For general values of $f_{\text{max}}$, we propose the following result that extends Theorem \ref{theorem1}:
\begin{theorem}\label{theorem2}
If $p(\cdot)$ is non-decreasing, the $Y_{i,j}$'s are i.i.d. with finite mean $\mathbb{E}[Y_{i,j}]<\infty$, the $X_{i,j}$'s are i.i.d. with finite mean $\mathbb{E}[X_{i,j}]<\infty$, the $Y_{i,j}$'s and the $X_{i,j}$'s are mutually independent, and Assumption \ref{ass1} holds, then \eqref{thm1-beta}-\eqref{root} is the optimal solution to \eqref{avg}, if the following condition holds:
\begin{equation}
\mathbb{E} \left[  X_{i,1} +Z_{i,1}(\beta)+Y' \right] >\frac{1}{f_{\text{max}}(1-\alpha)},
\end{equation}
where ${Y'} = Y_{i,1}+ \sum_{j=2}^{M_i} (X_{i,j}+Y_{i,j})$.
Otherwise, an optimal solution is as follows: 
\begin{align}
& Z_{i,1}(\beta) = \left\{
\begin{array}{lll}
  Z_{\text{min}} (\beta) & \text{w.p. $\lambda$,}\\
  Z_{\text{max}}(\beta)  & \text{w.p. $1-\lambda$. }
\end{array}
\right. \\ 
& Z_{i,j}=0, \ \ \ j = 2,3,\ldots,M_i,
\end{align}
$Z_{\text{min}}( \beta)$ and $Z_{\text{max}}( \beta)$ are described as follows:  
\begin{align}
& \nonumber Z_{\text{min}}( \beta)    =   \inf_{z}  \Big \{  z\ge 0:  \\
 & \mathbb{E}_{Y'} \left[ p( Y_{i-1,M_{i-1}}+X_{i,1} +z+Y')  \ \big{|} \  Y_{i-1,M_{i-1}}, X_{i,1} \right] \ge \beta \Big \}, \label{zmin} \\
& \nonumber Z_{\text{max}}( \beta)     =   \inf_{z}  \Big \{  z\ge 0: \\
 & \mathbb{E}_{Y'} \left[ p( Y_{i-1,M_{i-1}}+X_{i,1} +z+Y')  \ \big{|} \  Y_{i-1,M_{i-1}}, X_{i,1} \right] > \beta \Big \}. \label{zmax}
\end{align}
$\beta$ is determined by 
\begin{align} 
\nonumber & \mathbb{E} \left[ X_{i,1}+Z_{\text{min}}( \beta)+Y'   \right] \le \frac{1}{f_{\text{max}}(1-\alpha)}\\  \le & \mathbb{E} \left[ X_{i,1}+Z_{\text{max}}( \beta)+Y'   \right]. \label{theorem2-beta}
\end{align}
The probability $\lambda$ is given by 
\begin{equation}
\lambda = \frac{\mathbb{E} \left[ X_{i,1}+Z_{\text{max}}( \beta)+Y'   \right] -  \frac{1}{f_{\text{max}}(1-\alpha)} }{\mathbb{E} \left[ Z_{\text{max}} ( \beta) - Z_{\text{min}}( \beta) \right]}.
\end{equation}
\end{theorem}
\begin{proof}
See Section \ref{section-proof}.
\end{proof}
According to Theorem \ref{theorem2}, the proposed optimal policy may be randomized or deterministic. When $p(\cdot)$ is strictly increasing, we have $Z_{\text{min}}(\beta)=Z_{\text{max}}(\beta)$. Similar to Theorem~\ref{theorem1}, the optimal policy is stationary and deterministic in current age and previous backward delay. When $p(\cdot)$ is not strictly increasing, $Z_{\text{min}}(\beta)$ and $Z_{\text{max}}(\beta)$ may be different, so the optimal policy at $j=1$ is a random mixture of two deterministic sampling times. Note that when $Z_{\text{min}}(\beta)$ and $Z_{\text{max}}(\beta)$ may be different, the random optimal policy \emph{may be nonstationary}. In addition, we can solve \eqref{theorem2-beta} via low complexity algorithms such as bisection search. 

When $M_i=1$ (or $\alpha=0$) and $X_{i,j}=0$, Theorem \ref{theorem2} reduces to \cite[Theorem2]{sun2019sampling}. Combined with the discussions in Section \ref{no-sampling}, we conclude that our paper is an extension to some recent studies on sampling for optimizing age, e.g., \cite{yates2015lazy,sun2017update,sun2019sampling,tsai2020age,arafa2020timely,klugel2019aoi}.       

\section{Proof of the Main Result}\label{section-proof}

In this section, we provide the proof of our main results: Theorem \ref{theorem1} and Theorem \ref{theorem2}. In Section \ref{reformu}, we utilize the Lagrangian dual problem of the original long-term average problem and reformulate the Lagrangian dual problem into a per-epoch MDP problem. In Section \ref{per-epoc}, we solve the per-epoch MDP problem by formulating an exact optimal value function to the Bellman Equation, which is the key challenge to this paper. In Section \ref{proof-zeroduality}, we established zero duality gap to the Lagrangian problem, which ends our proof. Finally, in Section \ref{discussion}, we summarize our technical contribution and compare it with some related works.  

\subsection{Reformulation of Problem \eqref{avg}}\label{reformu}

In this subsection, we decompose the original problem to a per-epoch problem. The idea is motivated by recent studies that reformulate the average problem into a per-sample problem \cite{sun2017update,sun2019sampling,sun2017remote,ornee2019sampling,tsai2020age}.

Since $\{ S_{i,M_i} \}_i$ follows a regenerative process, by renewal theory, \cite[Section 6.1]{haas2006stochastic}, \cite{ross1996stochastic},
\begin{align}
&\limsup_{T\rightarrow \infty}  \frac{1}{T}  \mathbb{E} \left[  \int_{0}^{T} p(\Delta_t) dt \right] \label{ori-to-n} \\  = & \lim_{n\rightarrow \infty} \frac{\mathbb{E} \left[  \int^{D_{n,M_n}}_{0}  p(\Delta_t) dt \right] }{\mathbb{E} \left[  D_{n,M_n}  \right]}  \\
= &   \lim_{n\rightarrow \infty} \frac{ \sum_{i=1}^{n} \mathbb{E} \left[  \int^{D_{i,M_i}}_{D_{i-1,M_{i-1}}}  p(\Delta_t) dt \right] }{\sum_{i=1}^{n} \mathbb{E} \left[  D_{i,M_i} - D_{i-1,M_{i-1}}  \right]}.
\end{align}

In addition, 
\begin{align}
& \limsup_{T\rightarrow \infty} \frac{1}{T}  \mathbb{E} \left[ C(T) \right] = \lim_{n\rightarrow \infty}    \frac{ \mathbb{E} \left[  \sum_{i=1}^{n} M_i  \right] }{ \mathbb{E} \left[  S_{n,M_n}   \right]} \\ = & \lim_{n\rightarrow \infty}    \frac{n}{(1-\alpha) \mathbb{E} \left[  D_{n,M_n}   \right]}. \label{ori-to-nrate}
\end{align}

From \eqref{ori-to-n}-\eqref{ori-to-nrate}, the original problem \eqref{avg} is equivalent to 
\begin{align}
p_{\text{opt}} & = \inf_{\pi\in \Pi} \lim_{n\rightarrow \infty} \frac{ \sum_{i=1}^{n} \mathbb{E} \left[  \int^{D_{i,M_i}}_{D_{i-1,M_{i-1}}}  p(\Delta_t) dt \right] }{\sum_{i=1}^{n} \mathbb{E} \left[  D_{i,M_i} - D_{i-1,M_{i-1}}  \right]}, \label{fraction-multi}\\
& \text{s.t. } \lim_{n\rightarrow \infty} \frac{1}{n} \sum_{i=1}^{n}    \mathbb{E} \left[ D_{i,M_i} - D_{i-1,M_{i-1}}  \right] \ge \frac{1}{f_{\text{max}}(1-\alpha)}.
\end{align}

We consider the following MDP with a parameter $c\in \mathbb{R}$:
\begin{align}
\nonumber h(c) & =   \inf_{\pi\in \Pi} \lim_{n\rightarrow \infty} \frac{1}{n} \sum_{i=1}^{n}  \mathbb{E} \Bigg{[}  \int^{D_{i,M_i}}_{D_{i-1,M_{i-1}}}  p(\Delta_t) dt   \\  & \ \  -  c \left(  D_{i,M_i} - D_{i-1,M_{i-1}}  \right)   \Bigg{]}, \label{hc} \\
& \text{s.t.} \lim_{n\rightarrow \infty} \frac{1}{n} \sum_{i=1}^{n}    \mathbb{E} \left[ D_{i,M_i} - D_{i-1,M_{i-1}}  \right] \ge \frac{1}{f_{\text{max}}(1-\alpha)}.
\end{align}

By Dinkelbach's method \cite{dinkelbach1967nonlinear}, we have
\begin{lemma}\cite[lemma 2]{sun2019sampling} \label{dinklebach-lemma}

(i) $h(c) \lesseqqgtr 0$ if and only if $p_{\text{opt}} \lesseqqgtr c$.

(ii) The solution to \eqref{fraction-multi} and \eqref{hc} are equivalent.
\end{lemma}

We define the Lagrangian with $c=p_{\text{opt}}$:
\begin{align}
L(\pi; \gamma) & =  \lim_{n\rightarrow \infty} \frac{1}{n} \sum_{i=1}^{n}  \mathbb{E} \Bigg{[}  \int^{D_{i,M_i}}_{D_{i-1,M_{i-1}}}  p(\Delta_t) dt   \\ &  - (p_{\text{opt}}+\gamma) \left(  D_{i,M_i} - D_{i-1,M_{i-1}}  \right)   \Bigg{]} + \frac{\gamma}{f_{\text{max}(1-\alpha)}},
\end{align}
where $\gamma\ge 0$ is the dual variable. The primal problem is  
\begin{align}
l(\gamma) \triangleq \inf_{\pi\in \Pi} L(\pi; \gamma). \label{prime-problem}
\end{align}
The dual problem is 
\begin{align}
d \triangleq \max_{\gamma\ge 0} l(\gamma). \label{duel-problem}
\end{align}
Weak duality theorem \cite{bertsekas2003convex,boyd2004convex} implies that $d\le h(p_{\text{opt}})$. We will later show that the duality gap is $0$, i.e., $d = h(p_{\text{opt}})$.
Note that 
\begin{align}
& D_{i,M_i} - D_{i-1,M_{i-1}} =  \sum_{j=1}^{M_i} (X_{i,j}+Z_{i,j} + Y_{i,j} ),\\
& \mathbb{E} \left[  \int^{D_{i,M_i}}_{D_{i-1,M_{i-1}}}   p(\Delta_t) dt \right] \\ = & \mathbb{E} \left [   \int^{Y_{i-1,M_{i-1}}+  \sum_{j=1}^{M_i} (X_{i,j}+Z_{i,j} + Y_{i,j} ) }_{Y_{i-1,M_{i-1}}} p(t) dt \right].
\end{align}
Recall that the age decreases to $Y_{i-1,M_{i-1}}$ at time $D_{i-1,M_{i-1}}$. Note that $Y_{i-1,M_{i-1}}$ is independent of the history information by the sampling time $S_{i-1,M_{i-1}}$. Thus, the age evolution at the $i^{th}$ epoch is independent of the sampling decisions from the previous epochs $0,1,2,...,i-1$. Therefore, to solve \eqref{prime-problem}, minimizing each epoch separately is sufficient.   
We define the policy space $\Pi_i$ as the collection of sampling decisions $(Z_{i,1},Z_{i,2},...)$ at epoch $i$ such that the stochastic kernel 
\begin{align}\nonumber
Z_{i,j}(dz_{i,j} | y_{i-1,M_{i-1}}, x_{i,1}, z_{i,1}, y_{i,1}, \ldots, z_{i,j-1}, y_{i,j-1},x_{i,j} )
\end{align} is Borel measurable. The difference  between $\Pi_i$ and $\Pi$ is that the sampling decisions in $\Pi_i$ do not depend on the history information from previous epochs (except $Y_{i-1,M_{i-1}}$). Hence, it is easy to find that $\Pi_i\subset \Pi$.   

Then, by the analysis of the previous paragraph, we have the following result:
\begin{lemma}\label{lemma3}
An optimal solution to \eqref{prime-problem} satisfies
\begin{align}
\nonumber \inf_{ \pi \in \Pi_i} & \mathbb{E} \Bigg [   \int^{Y_{i-1,M_{i-1}}+  \sum_{j=1}^{M_i} (X_{i,j}+Z_{i,j} + Y_{i,j} ) }_{Y_{i-1,M_{i-1}}} p(t) dt \\ &  - (p_{\text{opt}} +\gamma) \sum_{j=1}^{M_i} (X_{i,j}+Z_{i,j} + Y_{i,j} )   \ \Big{|} \  Y_{i-1,M_{i-1}},X_{i,1} \Bigg ]. \label{per-epoch}
\end{align}
\end{lemma}

Thus, for any epoch $i$, we will solve $Z_{i,1}, Z_{i,2},...$ according to \eqref{per-epoch}.

\subsection{Solution to the Per-epoch Problem \eqref{per-epoch}}\label{per-epoc}

We will solve problem \eqref{per-epoch} given that $Y_{i-1,M_{i-1}}=\delta$ and $X_{i,1} = x$, where $\delta\ge0$ and $x\ge0$. Since the epoch number $i$ does not affect problem \eqref{per-epoch}, in this subsection, we will remove the subscription $i$ from $M_i,X_{i,j},Y_{i,j},Z_{i,j}$ and replace them by $M,X_j,Y_j,Z_j$ for the ease of descriptions. 
In addition, since we want to find out a solution to \eqref{fraction-multi}, we need to avoid that $l(\gamma)=-\infty$. Thus, we assume that $\gamma$ satisfies $\inf_{z\ge 0}\{z: p(z) > p_{\text{opt}} + \gamma \}<\infty$.\footnote{If $\inf_{z\ge 0}\{z: p(z) > p_{\text{opt}} + \gamma \}=\infty$, this subsection implies that waiting for arbitrary large time can optimize \eqref{prime-problem}. If such a policy optimizes \eqref{hc}, we have $p_{\text{opt}}=\bar{p}$, which contradicts to our assumption that $p_{\text{opt}}<\bar{p}$.} 

 Different from \cite{sun2017update,sun2017remote,ornee2019sampling,sun2019sampling}, the per-epoch problem \eqref{per-epoch} is an MDP with multiple samples and cannot be reduced to the per-sample problem in the sense that the age is not refreshed under failed transmissions. According to \eqref{per-epoch}, we define the value function $J_{\pi,\gamma}$ under a policy $\pi\in \Pi_i$ with an initial age state $\delta\ge 0$ (at delivery time) and backward delay $x\ge 0$: 
\begin{align} 
\nonumber  J_{\pi,\gamma}(\delta,x) = & \mathbb{E} \Bigg[   \int^{\delta+  \sum_{j=1}^{M} (X_{j}+Z_{j} + Y_{j} ) }_{\delta} p(t) dt \\ &   - (p_{\text{opt}} +\gamma ) \sum_{j=1}^{M} (X_{j}+Z_{j} + Y_{j} )  \ \Big{|} \   X_{1}=x   \Bigg]  \label{value-function-2}  \\
= & \mathbb{E} \left[   \sum_{j=1}^{M}  g_{\gamma}(\Delta_j, X_{j}, Z_{j})   \ \Big{|} \  \Delta_1 = \delta, X_{1}=x \right],  \label{value-function} 
\end{align} where the instant cost function $g_{\gamma}(\delta,x,z)$ with state $(\delta,x)$ and action $z$ is defined as 
\begin{align} 
\nonumber & g_{\gamma}(\delta,x,z) \\ = &  \mathbb{E}_Y \left[ 
 \int_{\delta}^{\delta+x+z+Y} p(t) dt - (p_{\text{opt}}+\gamma) (x+z+Y) \right], \label{state-evolve}
\end{align} 
where $Y$ has the same delay distribution as the $Y_j$'s 
and the age state evolution is described as
\begin{equation}\label{state-evolution}
\Delta_{j+1} = \Delta_{j}+X_j+Z_{j}+Y_{j}, \ \ \ j = 1,2,...M-1,
\end{equation} with initial age state $\Delta_1 = \delta$ and initial backward delay $x$. Also, the policy $\pi\in \Pi_i$ has a Borel measurable stochastic kernel $Z_{j}(dz_{j} |  \delta_1, x_{1}, z_{1}, \ldots, \delta_j, x_{j} )$, and thus $J_{\pi,\gamma}(\delta,x)$ is Borel measurable \cite[Chapter 9]{bertsekas2004stochastic}.
The above settings imply that problem \eqref{per-epoch} is equivalent to a shortest path MDP problem. 
Solving \eqref{per-epoch} is equivalent to solving 
\begin{equation} \label{optimal-value-function}
J_{\gamma}(\delta,x) = \inf_{\pi \in \Pi_i} J_{\pi,\gamma}(\delta,x).
\end{equation}

When the channel state is reliable, i.e., $\alpha=0$ or $M=1$, problem \eqref{per-epoch} (or equivalently, \eqref{optimal-value-function}) becomes a single-sample problem, and there is no bound restriction to the instant cost function $g_{\gamma}(\delta,x,z)$.
However, in the unreliable transmission case where $\alpha>0$, problem \eqref{per-epoch} contains multiple samples. In the case of multiple samples, most of the literature of dynamic programming e.g., \cite{bertsekas2018abstract,bertsekas1995dynamic1,bertsekas1995dynamic2,bertsekas2004stochastic,puterman2014markov,sennott1986new,sennott1989average,sennott1986neww} requires that the instant cost function $g_{\gamma}(\delta,x,z)$ is bounded from below. 
We have such a requirement.
\begin{lemma}\label{lemma_boundedbelow}
There exists a value $\eta$ such that $g_{\gamma}(\delta,x,z)\ge -\eta$ and $J_{\pi,\gamma}(\delta,x)\ge -\eta/(1-\alpha)$ for all $(\delta,x,z)$ and any policy $\pi\in\Pi_i$.  
\end{lemma}
\ifreport
\begin{proof}
See Appendix \ref{lemma_boundedbelowapp}. 
\end{proof}
\fi

Using Lemma \ref{lemma_boundedbelow}
\ifreport
and Appendix \ref{discountapp}, 
\fi
$J_{\pi,\gamma}(\delta,x)$ defined in \eqref{value-function} also equals to a discounted sum with discount factor $\alpha$: 
\begin{align}
 J_{\pi,\gamma}(\delta,x)
= & \sum_{j=1}^{\infty} \alpha^{j-1} \mathbb{E} \left[  g_{\gamma}(\Delta_j,X_j,Z_{j})   \ \Big{|} \  \Delta_1 = \delta, X_1 = x \right].   \label{discount}
\end{align}
Note that \eqref{discount} is motivated by \cite[Chapter 5]{bertsekas1995dynamic1}, illustrating that the discounted problem is equivalent to a special case of shortest path problem.

Recall that uncountable infimum of Borel measurable functions is not necessary Borel measurable. Problem \eqref{per-epoch} has an uncountable state space. Thus, the optimal value function $J_{\gamma}(\delta,x)$ defined in \eqref{optimal-value-function} \emph{may not be Borel measurable}\footnote{see \cite{bertsekas1995dynamic2,bertsekas2004stochastic} for counterexamples. In discrete-time system where the system time is slotted, we do not have this challenge.}, despite that $J_{\pi,\gamma}(\delta,x)$ is Borel measurable for all $\pi\in\Pi_i$. Then, some well known theories may not satisfy, such as the optimality of the Bellman equation among $\Pi_i$. One of the methods to overcome this challenge is to enlarge the policy spaces.    
We define a collection of policies $\Pi'_i$ such that the stochastic kernel $Z_{j}(dZ_{j} | \delta_1, x_1, z_1,\ldots, \delta_j,x_j )$ is universally measurable \cite{bertsekas2004stochastic}. Note that every Borel measurable stochastic kernel is a universally measurable stochastic kernel, so we have $\Pi_i\subset \Pi'_i$. 

Note that if $\pi\in \Pi'_i$, we also denote $J_{\pi,\gamma}(\delta,x)$ as the discounted cost of $\pi$ given in \eqref{discount}.
For all given age state $\delta$ and delay $x$, we define  
\begin{equation}\label{universe}
J'_{\gamma}(\delta,x) = \inf_{\pi\in \Pi'_i} J_{\pi,\gamma} (\delta,x).
\end{equation} 
It is easy to see that $J'_{\gamma}(\delta,x)\le J_{\gamma}(\delta,x)$. In this subsection, we will finally show that $J'_{\gamma}(\delta,x)= J_{\gamma}(\delta,x)$.

By Lemma \ref{lemma_boundedbelow}, it is easy to show that $J_{\pi,\gamma}\ge -\eta/(1-\alpha)$ for all $\pi\in \Pi'$. Using $J_{\pi,\gamma}\ge -\eta/(1-\alpha)$ and \cite[Corollary 9.4.1]{bertsekas2004stochastic}, $J'(\delta,x)$ is lower semianalytic \cite{bertsekas2004stochastic}. Note that any real-valued Borel measurable function is lower semianalytic. 
This allows us to consider the Bellman operator based on a general lower semianalytic function $u(\delta,x)$. For any deterministic and stationary policy $\pi\in \Pi_i$ with Borel measurable decisions $\pi(\delta,x)$,  
we define an operator $T_{\pi,\gamma}$ on a function $u$: 
\begin{align}
\nonumber & T_{\pi,\gamma} u(\delta,x) \\ = & g_{\gamma}(\delta,x,\pi(\delta,x)) +  \alpha \mathbb{E}_{Y,X} \left[  u(\delta+x+\pi(\delta,x)+Y,X) \right],
\end{align} 
where $Y$ and $X$ have the same distribution as the i.i.d. forward delay $Y_j$'s and backward delay $X_j$'s, respectively. 
We also define the Bellman operator $T_{\gamma}$ on the function $u$:      
\begin{equation}\label{bellman-operator}
T_{\gamma} u(\delta,x) = \inf_{z\in [0,\bar{z}]} g(\delta,x,z) + \alpha \mathbb{E}_{Y,X} \left[  u(\delta+x+z+Y,X) \right].
\end{equation}
As is described in Assumption \ref{ass1}, the bound $\bar{z}$ is taken sufficiently large. 
Note that if the function $u(\delta,x)$ is Borel measurable, $T_{\gamma} u(\delta,x)$ is not necessary Borel measurable in the sense that uncountable infimum of Borel measurable functions is not necessary Borel measurable. However, if we extend $u(\delta,x)$ to be lower semianalytic, then $T_{\gamma} u(\delta,x)$ is also lower semianalytic \cite[Proposition 7.47]{bertsekas2004stochastic}, i.e., $T_{\gamma}$ is well-defined under lower semianalytic functions. Note that the expectation on a lower semianalytic function has the same definition with the expectation on a Borel measurable function. In all, we have
\begin{lemma}\label{T-welldefine}
If $u(\delta,x)$ is lower semianalytic, then $T_{\pi,\gamma}u(\delta,x)$ and $T_{\gamma}u(\delta,x)$ are both lower semianalytic. 
\end{lemma}
\begin{proof}
\ifreport
See Appendix \ref{T-welldefineapp}.
\fi
\end{proof}

We denote $u_1 = u_2$ if $u_1(\delta,x) = u_2(\delta,x)$ for all $\delta,x \in [0,\infty)$. Using the definition of $T_{\pi,\gamma}$ and $T_{\gamma}$, the discounted problem \eqref{universe} has the following properties \cite[Chapter 9.4]{bertsekas2004stochastic}:
\begin{lemma}\label{bellman-optimal}
If $p(\cdot)$ is non-decreasing, the $Y_{j}$'s are i.i.d. with finite mean $\mathbb{E}[Y_{j}]<\infty$, the $X_{j}$'s are i.i.d. with finite mean $\mathbb{E}[X_{j}]<\infty$, the $Y_{j}$'s and the $X_{j}$'s are mutually independent, then the optimal value function $J'_{\gamma}(\delta,x)$ defined in \eqref{universe} satisfies the Bellman equation:
\begin{equation}\label{bellman}
J'_{\gamma} = TJ'_{\gamma},
\end{equation} i.e., the optimal value function $J'_{\gamma}$ is a fixed point of $T_{\gamma}$.
\end{lemma}

To derive an optimal policy, we first provide two stationary and deterministic policies called $\mu_{\text{min},\gamma}$ and $\mu_{\text{max},\gamma}$. Then we will show that both $\mu_{\text{min},\gamma}$ and $\mu_{\text{max},\gamma}$ are the solution to problem \eqref{per-epoch}. 
\begin{definition}\label{definition-mu}
The stationary and deterministic policies $\mu_{\text{min},\gamma}$ and $\mu_{\text{max},\gamma}$ are defined as 
\begin{align}
&\mu_{\text{min},\gamma}(\delta,x) = \max\{  b_{\text{min},\gamma}-\delta-x,0   \}, \label{def-of-mu}  \\ 
&\mu_{\text{max},\gamma}(\delta,x) = \max\{  b_{\text{max},\gamma}-\delta-x,0   \}, \label{def-of-mu2}  \\ 
&  b_{\text{min},\gamma} = \inf_{c} \left\{ c\ge0: \mathbb{E} \left[ p(c+Y')  \right] \ge p_{\text{opt}} + \gamma \right\}, \label{const-b} \\
&   b_{\text{max},\gamma} = \inf_{c} \left\{ c\ge0: \mathbb{E} \left[ p(c+Y')  \right] > p_{\text{opt}} + \gamma \right\},  \\
& Y' \triangleq Y_1+\sum_{j=2}^{M}(X_j+Y_{j}).
\end{align}
A randomized policy $\tilde{\mu}_{\lambda,\gamma} = \{ Z_1,Z_2,... \}$ with $\lambda\in [0,1]$ satisfies
\begin{align}
& Z_{1} = \left\{
\begin{array}{lll}
  \mu_{\text{min},\gamma}(\delta,x) & \text{w.p. $\lambda$,}\\
  \mu_{\text{max},\gamma}(\delta,x)  & \text{w.p. $1-\lambda$. }
\end{array}
\right. \\ 
& Z_{j}=0, \ \ \ j = 2,3,...,M_i.
\end{align}
\end{definition}
\noindent Using the definition of $\Pi_i$, we have $\mu_{\text{min},\gamma}, \mu_{\text{max},\gamma}\in \Pi_i$, and $\tilde{\mu}_{\lambda,\gamma}\in \Pi_i$ for all $\lambda\in [0,1]$\cite[Chapter 7]{bertsekas2004stochastic}.

Upon delivery of the first sample, age of $\mu_{\text{min},\gamma}$ and $\mu_{\text{max},\gamma}$ increase to $\Delta_2 = \delta+x+\mu_{\text{min},\gamma}(\delta,x)+Y_{1}$ and $\delta+x+\mu_{\text{max},\gamma}(\delta,x)+Y_{1}$, which are larger than $ \max \{ \delta+x, b_{\text{min},\gamma} \}$, $ \max \{ \delta+x, b_{\text{max},\gamma} \}$, respectively. 
Then, the waiting time for the second sample is $\mu_{\text{min},\gamma}(\Delta_2,X_2)=0$ and $\mu_{\text{max},\gamma}(\Delta_2,X_2)=0$, respectively. Thus, the waiting time at stage $2$,... is $0$ under $\mu_{\text{min},\gamma}$ and $\mu_{\text{max},\gamma}$. Therefore, we have $\mu_{\text{min},\gamma} = \tilde{\mu}_{0,\gamma}$ and $\mu_{\text{max},\gamma} = \tilde{\mu}_{1,\gamma}$.
Note that when we do not consider sampling rate constraint, then $\gamma=0$, and the policy $\mu_{\text{min},\gamma}$ is equivalent to \eqref{thm1-beta} and \eqref{thm1-beta2} in Theorem \ref{theorem1}.
It remains to show that $\mu_{\text{min},\gamma}$ and $\mu_{\text{max},\gamma}$ are indeed optimal to problem \eqref{per-epoch}. 

Recall that we denote $J_{\pi,\gamma} (\delta,x)$ to be the value function with initial state $\delta,x$ under a policy $\pi$. 
Then, we have the following key result:

\begin{lemma}\label{prop-optimal}
If $p(\cdot)$ is non-decreasing, the $Y_{j}$'s are i.i.d. with finite mean $\mathbb{E}[Y_{j}]<\infty$, the $X_{j}$'s are i.i.d. with finite mean $\mathbb{E}[X_{j}]<\infty$, the $Y_{j}$'s and the $X_{j}$'s are mutually independent, then  
the value functions $J_{\mu_{\text{min},\gamma} }(\delta,x)$ and $J_{\mu_{\text{max},\gamma} }(\delta,x)$ 
satisfy
\begin{equation}
J_{\mu_{\text{min},\gamma} } = T_{\gamma} J_{\mu_{\text{min},\gamma} } = J_{\mu_{\text{max},\gamma} } = T_{\gamma} J_{\mu_{\text{max},\gamma} }.
\end{equation}
Moreover, for any $\lambda\in [0,1]$, we have $J_{\tilde{\mu}_{\lambda,\gamma}} = J_{\mu_{\text{min},\gamma} } = J_{\mu_{\text{max},\gamma} }$.
\end{lemma}

\begin{proof}
We provide the proof sketch of $J_{\mu_{\text{min},\gamma} } = T_{\gamma} J_{\mu_{\text{min},\gamma} } $ here and replace $\mu_{\text{min},\gamma}$ by $\mu$ for simplicity. We relegate the detailed proof
\ifreport
in Appendix \ref{prop-optimalapp}.
\fi

We define the q-function $Q_{\mu}(\delta,x,z)$ as the cost of starting at state $(\delta,x)$, waiting for time $z$ for the first sample, and then following policy $\mu$ for the remaining samples \cite[Section 6]{bertsekas1995dynamic1}\cite[Chapter 3]{sutton2018reinforcement}. It is easy to find that  
\begin{equation}\label{q-function-2}
Q_{\mu}(\delta,x,z) = g_{\gamma}(\delta,x,z) + \alpha \mathbb{E} \left[   J_{\mu}(\delta+x+z+Y,X) \right].
\end{equation}
From \eqref{q-function-2} and \eqref{bellman-operator}, showing $J_\mu = T_{\gamma} J_\mu$ is equivalent to showing that $Q_{\mu}(\delta,x,z) \ge J_{\mu}(\delta,x)$ for all the waiting time $z \ge 0$, age $\delta$ and $x$.
We have stated that $\mu$ has a nice structure: for any initial state, the waiting times of stage $2,3...$ are $0$. Thus, we can derive the closed form expression of $J_{\mu}(\delta,x)$ according to \eqref{value-function-2} (where $Z_{j}=0$ for $j\ge2$). Also, given the definition of $Q_{\mu}(\delta,x,z)$, we can derive the expression of $Q_{\mu}(\delta,x,z)$ with the similar form of \eqref{value-function-2}. 
By comparing $Q_{\mu}(\delta,x,z)$ and $J_{\mu}(\delta,x)$, 
we can finally show that $Q_{\mu}(\delta,x,z) \ge J_{\mu}(\delta,x)$ for all $(\delta,x,z)$.
\end{proof}

Lemma \ref{prop-optimal} tells that $J_{\mu_{\text{min},\gamma}}$ (or equivalently, $J_{\mu_{\text{max},\gamma}}$) is a fixed point of $T_\gamma$. From Lemma \ref{bellman-optimal}, the optimal value function $J'_{\gamma}$ is also a fixed point of $T_\gamma$. To show that $J'_{\gamma} =J_{ \mu_{\text{min},\gamma}}$, it remains to show that the fixed point of $T_\gamma$ is unique. 
If the age penalty $p(\cdot)$ is bounded, $J_{\pi,\gamma}(\delta,x)$ is bounded for any policy $\pi\in \Pi_i'$. Then, according to the contraction mapping theorem, the bellman equation \eqref{bellman} has a unique bounded solution \cite{bertsekas1995dynamic2,ross1996stochastic,sennott1986neww}, i.e., $J_{ \mu_{\text{min},\gamma}}=J'_{\gamma}$. 
Note that there may be unbounded solutions to \eqref{bellman} \cite{sennott1986new,sennott1989average}.  
If $p(\cdot)$ is unbounded, we will utilize Assumption \ref{ass1} to show the uniqueness. 

Let us denote $\Lambda = [0,\infty)\times [0,\bar{x}]$, where $\bar{x}$ is the bound of $X_j$ mentioned in Assumption \ref{ass1}. In Assumption \ref{ass1}, we have defined an increasing function $v(\delta): [0,\infty) \rightarrow \mathbb{R}^+$ (also called the \emph{weighted function}). 
The \emph{weighted sup-norm} $\|u\|$ of a function $u:\Lambda \rightarrow \mathbb{R}$ is defined as
\begin{equation}\label{def-norm}
\| u \|  = \max_{(\delta, x) \in \Lambda} \frac{|u(\delta, x)|}{v(\delta)}.
\end{equation}
Let $B(\Lambda)$ denote the set of all lower semianalytic functions $u:\Lambda \rightarrow \mathbb{R}$ such that $\| u \|<\infty$.
Note that any real-valued Borel measurable function is lower semianalytic. 
From \cite[p. 47]{bertsekas1995dynamic2}, \cite[Lemma 7.30.2]{bertsekas2004stochastic}, $B(\Lambda)$ is complete under the weighted sup-norm.    
\begin{lemma}\label{preserves}
If $p(\cdot)$ is non-decreasing, the $Y_{j}$'s are i.i.d. with finite mean $\mathbb{E}[Y_{j}]<\infty$, the $X_{j}$'s are i.i.d. with finite mean $\mathbb{E}[X_{j}]<\infty$, the $Y_{j}$'s and the $X_{j}$'s are mutually independent, and Assumption \ref{ass1} holds, then for all $\pi\in \Pi_i$, $J_{\pi,\gamma} \in B(\Lambda)$. 
\end{lemma}
\ifreport
\begin{proof}
See Appendix \ref{preservesapp}.
\end{proof}
\fi

Then, the following result shows the uniqueness of the Bellman equation $T_\gamma u=u$. 
\begin{lemma}\label{contraction-unique}
If $p(\cdot)$ is non-decreasing, the $Y_{j}$'s are i.i.d. with finite mean $\mathbb{E}[Y_{j}]<\infty$, the $X_{j}$'s are i.i.d. with finite mean $\mathbb{E}[X_{j}]<\infty$, the $Y_{j}$'s and the $X_{j}$'s are mutually independent, and Assumption \ref{ass1} holds, the following conditions hold:

(a) For any lower semianalytic function $u:\Lambda\to \mathbb{R}$, if $u\in B(\Lambda)$, then $T_{\pi,\gamma} u \in B(\Lambda)$ for all deterministic and stationary policy $\pi\in \Pi_i$, and $T_\gamma u\in B(\Lambda)$.

(b) The Bellman operator $T_\gamma$ has an $m$-stage contraction mapping with modulus $\rho$, i.e., for all $u_1,u_2\in B(\Lambda)$,  
\begin{equation}
\| T^m_\gamma u_1 - T^m_\gamma u_2 \| \le \rho \| u_1-u_2 \|,
\end{equation}
where constants $\rho\in(0,1)$ and $m$ are mentioned in Assumption \ref{ass1}, and the weighted sup-norm $\| \cdot \|$ is defined in \eqref{def-norm}. 

(c) There exists a unique function $u\in B(\Lambda)$ such that $T_\gamma u = u$. 
\end{lemma}
\ifreport
\begin{proof}
See Appendix \ref{contraction-uniqueapp}. 
\end{proof}
\fi

 



From Lemma \ref{preserves}, $J_{\mu_{\text{min},\gamma}} \in B(\Lambda)$. From Lemma \ref{contraction-unique}(c), Lemma~\ref{prop-optimal} and $J_{\mu_{\text{min},\gamma}} \in B(\Lambda)$, $J_{\mu_{\text{min},\gamma}}$ (or equivalently, $J_{\mu_{\text{max},\gamma}}$) is the unique solution to $T_\gamma u = u$. From Lemma \ref{bellman-optimal}, $J_{\mu_{\text{min},\gamma}} = J'_{\gamma}$. Since $\mu_{\text{min},\gamma},\mu_{\text{max},\gamma} \in \Pi_i$ and $\Pi_i\subset \Pi'_i$,
 $\mu_{\text{min},\gamma}$ and $\mu_{\text{max},\gamma}$ are the optimal policies in $\Pi_i$. Note that $\mu_{\text{min},\gamma}= \tilde{\mu}_{0,\gamma}$ and $\mu_{\text{max},\gamma}= \tilde{\mu}_{1,\gamma}$. Using Lemma \ref{prop-optimal}, we immediately get the final result:

\begin{lemma}\label{prime-lemma}
A collection of optimal policies to problem \eqref{per-epoch} is $\{ \tilde{\mu}_{\lambda,\gamma}: \lambda\in [0,1] \}$ described in Definition \ref{definition-mu}.
\end{lemma}


\subsection{Optimal Solution to \eqref{hc} When $c=p_{\text{opt}}$}\label{proof-zeroduality}
Section \ref{per-epoc} provides the optimal solution to \eqref{per-epoch} given the initial states $Y_{i-1,M_{i-1}}=\delta$ and $X_{i,1}=x$. 
Using Lemma \ref{prime-lemma} and strong duality, we have the following result.

\begin{theorem}\label{theorem3}
If $p(\cdot)$ is non-decreasing, the $Y_{i,j}$'s are i.i.d. with finite mean $\mathbb{E}[Y_{i,j}]<\infty$, the $X_{i,j}$'s are i.i.d. with finite mean $\mathbb{E}[X_{i,j}]<\infty$, the $Y_{i,j}$'s and the $X_{i,j}$'s are mutually independent, and Assumption~\ref{ass1} holds, then $\mu_{\text{min},0}$ described in Definition \ref{definition-mu} is an optimal solution to \eqref{hc} with $c=p_{\text{opt}}$, if the following condition holds:
\begin{equation}
\mathbb{E} \left[  X_{i,1} +\mu_{\text{min},0}(Y_{i-1,M_{i-1}},X_{i,1})+Y' \right] >\frac{1}{f_{\text{max}}(1-\alpha)}, \label{rate-useless}
\end{equation} 
where ${Y'} = Y_{i,1}+ \sum_{j=2}^{M_i} (X_{i,j}+Y_{i,j})$.
Otherwise, $\tilde{\mu}_{\lambda,\gamma}$ is an optimal solution to \eqref{hc} with $c=p_{\text{opt}}$, where $\gamma$ is determined by 
\begin{align} 
\nonumber & \mathbb{E} \left[ X_{i,1}+\mu_{\text{min},\gamma}(Y_{i-1,M_{i-1}},X_{i,1})+Y'   \right] \le \frac{1}{f_{\text{max}}(1-\alpha)}\\  \le & \mathbb{E} \left[ X_{i,1}+\mu_{\text{max},\gamma}(Y_{i-1,M_{i-1}},X_{i,1})+Y'   \right], 
\end{align}
and the probability $\lambda$ is given by 
\begin{equation}
\lambda = \frac{\mathbb{E} \left[ X_{i,1}+\mu_{\text{max},\gamma}(Y_{i-1,M_{i-1}},X_{i,1})+Y'   \right] -  \frac{1}{f_{\text{max}}(1-\alpha)} }{\mathbb{E} \left[ \mu_{\text{max},\gamma}(Y_{i-1,M_{i-1}},X_{i,1}) - \mu_{\text{min},\gamma}(Y_{i-1,M_{i-1}},X_{i,1}) \right]}.\label{theorem3-prob}
\end{equation}
\end{theorem}
\ifreport
\begin{proof}
See Appendix \ref{zerodualityapp}.
\end{proof}
\fi
By taking $\beta=p_{\text{opt}}+\gamma$,
Theorem $\ref{theorem2}$ is directly shown by Theorem \ref{theorem3}.

In addition, note that Theorem \ref{theorem1} is directly shown by Lemma \ref{prime-lemma}, by taking $\beta = p_{\text{opt}}$ and $\gamma=0$. In other words, $\mu_{\text{min},0}$ is an optimal solution to \eqref{avg} when $f_{\text{max}}=\infty$.

\subsection{Discussion}\label{discussion}
Many existing studies on AoI sampling assume that the transmission channel is error-free, i.e., $M_i=1$ for all $i$, e.g.,  \cite{yates2015lazy,sun2017update,sun2019sampling,tsai2020age,tsai2020unifying,sun2017remote,ornee2019sampling}. Due to the renewal property, their original problems are reduced to a per-sample problem. Similarly, our result is equivalent to the per-epoch problem illustrated in \eqref{per-epoch}. If $M_i=1$, problem \eqref{per-epoch} reduces to a per-sample problem, where there is only one decision $Z_{i,1}$ and is solved using convex optimization. However, when $M_i\ne 1$, problem \eqref{per-epoch} is an MDP that contains multiple samples. This MDP cannot be solved by convex optimization (e.g., \cite{yates2015lazy,sun2017update,sun2019sampling,tsai2020age}) or optimal stopping rules (e.g., \cite{tsai2020unifying,sun2017remote,ornee2019sampling}).

Therefore, one of the technical contributions in this paper is to accurately solve the MDP in \eqref{per-epoch}. We summarize the high-level idea of solving \eqref{per-epoch}: First, in Lemma \ref{bellman-optimal}, among the extended policy space $\Pi'_i$ with universally measurable stochastic kernel \cite[Chapter 7]{bertsekas2004stochastic}, the optimal policy  satisfies the Bellman equation \eqref{bellman}. Then, in Lemma \ref{prop-optimal}, we provide the exact value function that is the solution to the Bellman Equation. Finally, under Assumption \ref{ass1} and Lemma \ref{contraction-unique}, the uniqueness of the Bellman equation is guaranteed.   

In addition, although we focus on continuous-time systems in this paper, our results can be easily reduced to the discrete-time systems by removing the content of measure theory.

\begin{figure}[t]
\includegraphics[scale=.36]{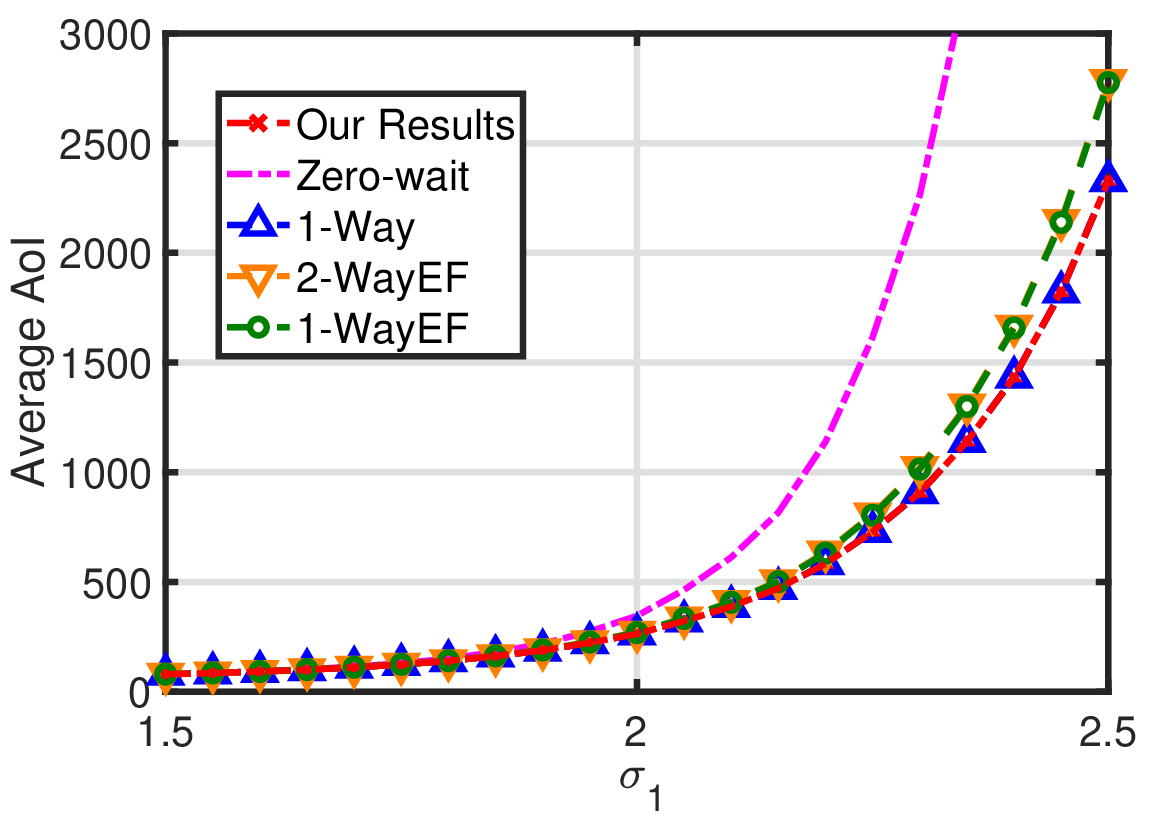}
\centering
\captionsetup{justification=justified}
\caption{Average AoI versus the parameter $\sigma_1$ of the forward channel, where $\sigma_2=1.5$ and $\alpha=0.8$.}
 \label{linearsigma1}
\end{figure}
\begin{figure}[t]
\includegraphics[scale=.37]{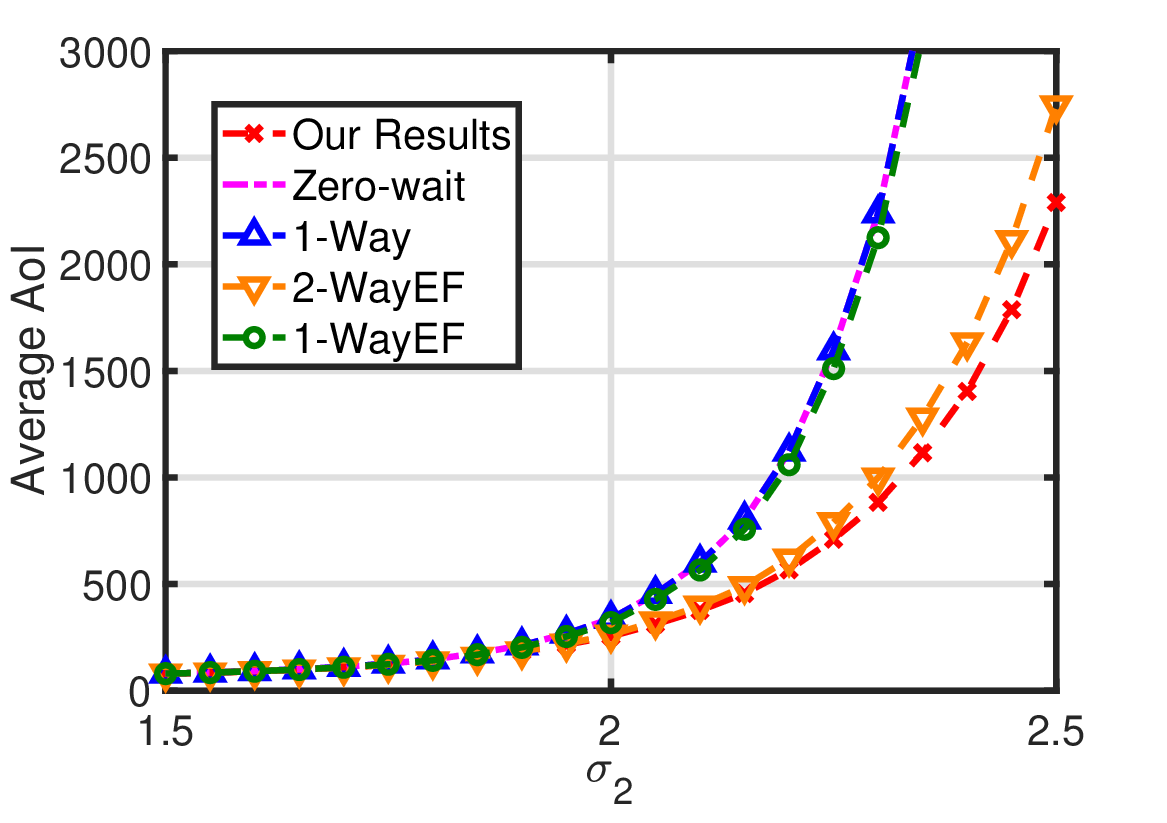}
\centering
\captionsetup{justification=justified}
\caption{Average AoI versus the parameter $\sigma_2$ of the backward channel, where $\sigma_1=1.5$ and $\alpha=0.8$.}
 \label{linearsigma2}
\end{figure}
\begin{figure}[t]
\includegraphics[scale=.4]{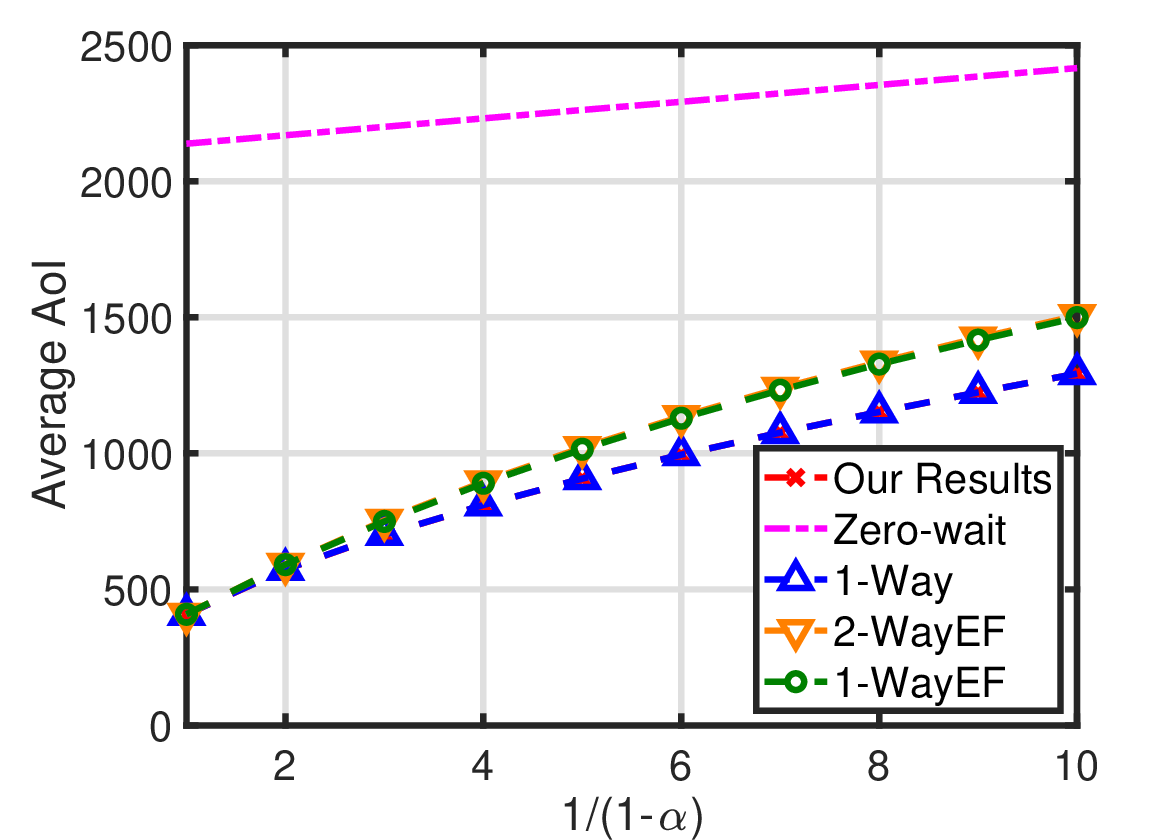}
\centering
\captionsetup{justification=justified}
\caption{Average AoI versus $1/(1-\alpha)$, where $\sigma_1 = 2.3$ and $\sigma_2 = 1.5$.}
 \label{linearp1}
\end{figure}
\begin{figure}[t]
\includegraphics[scale=.35]{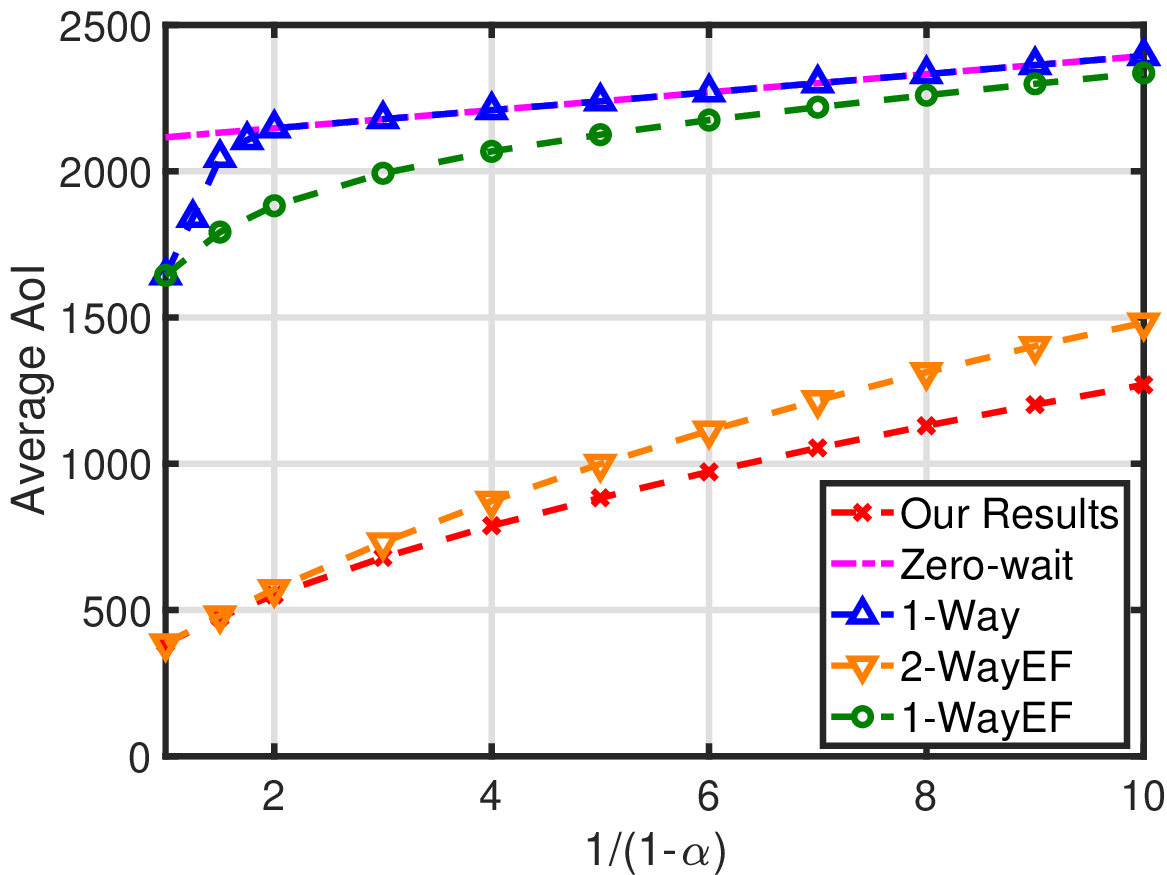}
\centering
\captionsetup{justification=justified}
\caption{Average AoI versus $1/(1-\alpha)$, where $\sigma_1=1.5$ and $\sigma_2=2.3$.}
 \label{linearp2}
\end{figure}

\section{Numerical Results}\label{numerical}

In this section, we compare our optimal sampling policy with the following sampling policies:

$1$. Zero-wait: Let $Z_{i,j}=0$, i.e., the source transmits a sample once it receives the feedback.   

$2$.  One-way ($1$-way): It falsely assumes that  the backward delay $X_{i,j}=0$ despite that $X_{i,j}$ may not be zero. 

$3$. Two-way Error-free ($2$-wayEF) \cite{tsai2020age}: It assumes that the forward channel's probability of failure $\alpha=0$ despite that $\alpha$ may not be zero. 

$4$. One-way Error-free ($1$-wayEF) \cite{sun2019sampling}: It assumes that $X_{i,j}=0$ and $\alpha=0$.

In this section, we consider linear age penalty $p(\delta)=2\delta$ and lognormal distributions on both forward and backward delay with scale parameters $\sigma_1, \sigma_2$, respectively. Note that the lognormal random variable with scale parameter $\sigma$ is expressed as $e^{\sigma R}$, where $R$ is the standard normal random variable.  The numerical results below show that our proposed policy always achieves the lowest average age.

Fig. \ref{linearsigma1} and Fig. \ref{linearsigma2}  illustrate the relationship between age and $\sigma_1,\sigma_2$, respectively. In Fig. \ref{linearsigma1}, we plot the evolution of average age in $\sigma_1$ given that $\sigma_2 = 1.5$ and $\alpha=0.8$. As $\sigma_1$ increases, the lognormal distribution of the forward channel becomes more heavy tailed. We observe that Zero-wait policy evolves much quicker than other policies in $\sigma_1$. In addition, $2$-wayEF and $1$-wayEF policies grow faster than the optimal policy in $\sigma_1$. In Fig. \ref{linearsigma2}, we fix $\sigma_1=1.5$ and plot the average age of the listed policies in $\sigma_2$. Unlike Fig. \ref{linearsigma1}, $1$-way and $1$-wayEF policies perform poorly since they fail to take highly random backward delay into account.

Fig. \ref{linearp1} and Fig. \ref{linearp2} depict the evolution of average age in $1/(1-\alpha)$, where $(\sigma_1,\sigma_2)=(2.3,1.5)$ and 
 $(\sigma_1,\sigma_2)=(1.5,2.3)$, respectively. Note that $1/(1-\alpha)$ is the average number of samples attempted for a successful transmission. In Fig. \ref{linearp1} and Fig. \ref{linearp2}, when $1/(1-\alpha)$ increases,
the gap between $2$-wayEF policy and our optimal policy increases. In Fig. \ref{linearp2}, since $\sigma_2>\sigma_1$, the tail of backward delay is heavier than that of forward delay. Thus, $1$-way and $1$-wayEF, which neglect the knowledge of backward delay, fail to improve the age performance. 
 
In summary, when either one of the channels is highly random, (i) Zero-wait policy is far from optimal, (ii) the age performance of $1$-wayEF or $2$-wayEF policy gets worse if the forward channel is more unreliable, (iii) $1$-way and $1$-wayEF polices are far from optimal if the backward channel is highly random.  
\vspace{-0.1cm}
\section{Conclusion}\label{conclusion}
In this paper, we design a sampling policy to optimize data freshness, where the source generates the samples and sends to the remote destination via a fading forward channel, and the acknowledgements are sent back via a backward channel. 
We overcome the curse of dimensionality that arises from the time-varying forward channel conditions and the randomness of the channel delays in both directions.
We reveal that the optimal sampling policy has a simple threshold based structure, and the optimal threshold is computed efficiently. 
\bibliographystyle{ieeetr}
\bibliography{sample}
\ifreport
\appendices

\section{Proof of \eqref{hat_ot}}\label{hat_otapp}

At time $t \in [D_{i,M_i},D_{i+1,M_{i+1}})$, the estimator has received the following information: (i) the sequence of the source process $\{O_{S_{j,M_j}}\}_{j\le i}$, (ii) the linear observations $\{B_{\tau}\}_{0\le \tau \le t}$, and (iii) the causal information of the channel delays that are prior to $D_{i,M_i}$, denoted as $\mathcal{H}^{remote}_{t}$ for simplicity. Then, the MMSE $\hat{O}_t$ satisfies
\begin{align}
\nonumber \hat{O}_t =  \mathbb{E} & \Big[ O_t \Big{|} \{ B_\tau \}_{0\le \tau < S_{i,M_i} }, \{ B_\tau \}_{S_{i,M_i}\le \tau\le t}, \{O_{S_{j,M_j}}\}_{j\le i}, \\  & \mathcal{H}^{remote}_t, D_{i,M_i}, D_{i+1,M_{i+1}} \Big]. 
\end{align}
Using the strong Markov property of $O_t$ and the assumption that channel delays and the processes $W_t$, $V_t$ are independent with $O_t$, we have 
\begin{align}
 \hat{O}_t =  \mathbb{E}  \Big[ O_t \Big{|} \{ B_\tau \}_{S_{i,M_i}\le \tau \le t}, O_{S_{i,M_i}}, D_{i,M_i}, D_{i+1,M_{i+1}} \Big]. 
\end{align}
Since in this paper, we assume that the sampling decision is independent of $O_t$, $O_t$ has no correlation with $D_{i,M_i}, D_{i+1,M_{i+1}}$. Therefore,  
\begin{align}
 \hat{O}_t =  \mathbb{E}  \Big[ O_t \Big{|} \{ B_\tau \}_{S_{i,M_i}\le \tau \le t}, O_{S_{i,M_i}} \Big], 
\end{align} which is the same as \eqref{hat_ot}.

\section{Proof of Proposition \ref{lim1}}\label{lim1app}

For any matrix $\bm{N}$, we denote $[\bm{N}]_{i,j}$ as the $i$th row and $j$th column element of $\bm{N}$. Similarly, for any vector $X$, $[X]_j$ is the $j$th element of $X$. We denote $\bm{N} \ge 0$ if $\bm{N}$ is positive semidefinite. For a matrix function $\bm{N}_t$, we denote $d\bm{N}_t/dt$ as the matrix that takes derivation on $t$ in each element of $\bm{N}_t$. Also, we say $\bm{N}_t$ is non-decreasing in $t$ if for any $a<b$, $\bm{N}_b - \bm{N}_a \ge 0$. Note that for any two matrices $\bm{N}_a$ and $\bm{N}_b$ in $\mathbb{R}^{n\times n}$, if $\bm{N}_b - \bm{N}_a \ge 0$, then $tr(\bm{N}_b)\ge tr(\bm{N}_a)$.

According to \textcolor{black}{\cite[Proposition VII.C.2]{poor2013introduction}}, $\bm{N}_t$ for $t
\in[D_{i,M_i},D_{i+1,M_{i+1}})$ satisfies the following Riccati differential equation:
\begin{align}
\frac{d\bm{N}_t}{dt}=\bm{\Theta N}_t + \bm{N}_t \bm{\Theta}^T +\bm{ \Sigma\Sigma}^T - \bm{N}_t \bm{H}^T \bm{R}^{-1} \bm{H} \bm{N}_t. \label{rec_eqnmulti}
\end{align}
Note that $\bm{N}_{S_{i,M_i}}=\bm{0}_{n\times n}$. Therefore, $\bm{N}_t$ is computed completely from time $S_{i,M_i}$ to $t$, provided some constant parameter matrices. Note that the age is defined as $\Delta_t = t - S_{i,M_i}$, so $\bm{N}_t$ along with $tr(\bm{N}_t)$ are a function of age $\Delta_t$.  

It remains to show that $tr(\bm{N}_t)$ is non-decreasing in $\Delta_t$. Since $\bm{N}_{S_{i,M_i}}=\bm{0}_{n\times n}$, by \eqref{rec_eqnmulti}, it is easy to see that $d\bm{N}_t/dt\ge 0$ when $t = S_{i,M_i}$. 
By using \cite[theorem3]{poubelle1988fake}, $d\bm{N}_t/dt\ge 0$ for $t\ge S_{i,M_i}$, and $\bm{N}_t$ is non-decreasing in $t$ for $t\ge S_{i,M_i}$. Thus, we conclude that the MMSE $tr(\bm{N}_t)$ is a \emph{non-decreasing} function of the age $\Delta_t$. 

\section{Proof of Proposition \ref{lemma-lim2}}\label{lemma-lim2app}

In one-dimensional case, the MMSE is equal to $n_t$ and the differential equation \eqref{rec_eqnmulti} reduces to
\begin{align}
\frac{dn_t}{dt}=-2\theta n_t+\sigma^2-\frac{h^2}{r}n_t^2.\label{rec_eqn}
\end{align}
Define $\mu_t = n_t - \bar{n}$, then $\mu_t$ satisfies 
\begin{align}
\frac{d\mu_t}{dt}=-(2\theta +2\frac{h^2}{r}\bar{n}) \mu_t-\frac{h^2}{r}\mu_t^2.\label{rec_eqn3}
\end{align}
The differential Equation \eqref{rec_eqn3} is a Bernoulli equation and the closed-from solution to \eqref{Lem1eq} can be derived accordingly, with the initial condition $n_{S_{i,M_i}}=0$, for $t\in[D_{i,M_i},D_{i+1,M_{i+1}})$. Therefore, in one-dimensional case, we can further solve the MMSE $n_t$ in closed-form, which is proved below. 
 


For simplicity, let us define $A=-\frac{h^2}{r}$, $B=-2\theta$, and $C=\sigma^2$. Note that we only use the above definitions of $A,B,C$ in this subsection (i.e., Appendix \ref{lim1app}). Hence, \eqref{rec_eqn} is equivalent to
\begin{align}
\frac{dn_t}{dt}=An_t^2+B n_t+C.\label{rec_eq12}
\end{align}
We define a function $\mu_t$ such that
\begin{align}
n_t=\bar{n}+\mu_t,\label{A_eq13}
\end{align}
where $\mu_t$ is a function of time $t$. Since $\bar{n}$ is a constant, we have
$dn_t/dt=d\mu_t/dt$.
We get
\begin{equation}
\frac{d\mu_t}{dt}=A\bar{n}^2+B\bar{n}+C+2A\bar{n}\mu_t+A\mu_t^2+B\mu_t.
\end{equation}
Since $\bar{n}$ is the steady state solution of \eqref{rec_eq12}, we have $A\bar{n}^2+B\bar{n}+C=0$. Hence, we have 
\begin{equation}\label{Ut+funct}
\frac{d\mu_t}{dt}=(2A\bar{n}+B)\mu_t+A\mu_t^2.
\end{equation}
Define $\eta_t=1/\mu_t$. This implies that $d\mu_t/dt=-(1/\eta_t^2) d\eta_t/dt$. Substitute by this in \eqref{Ut+funct}, we get
\begin{align}
& \frac{d\eta_t}{dt} = -(2A\bar{n}+B)\eta_t-A,  \\
\Longleftrightarrow & \frac{d (\eta_t+\frac{A}{2A\bar{n}+B})}{dt} = -(2A\bar{n}+B)(\eta_t+\frac{A}{2A\bar{n}+B}).\label{de_eq17}
\end{align}
The general solution to \eqref{de_eq17} is 
\begin{align}
\eta_t=\frac{-A}{2A\bar{n}+B}+Ke^{-(2A\bar{n}+B)(t-S_{i,M_i})},
\end{align}
with a constant $K$. Substitute this back into $\mu_t$, and then in \eqref{A_eq13}, we get
\begin{align}
n_t=\bar{n}+\frac{1}{\frac{-A}{2A\bar{n}+B}+Ke^{-(2A\bar{n}+B)(t-S_{i,M_i})}}.\label{Ap_eq21}
\end{align} 
To determine $K$, we substitute by our initial condition $n_{S_i}=0$. This implies 
\begin{align}
K=\frac{A}{2A\bar{n}+B}-\frac{1}{\bar{n}} = l-\frac{1}{\bar{n}}, \label{remote_10}
\end{align} where $l$ is defined in \eqref{remote_9}
Substitute by this back into \eqref{Ap_eq21}, observing that $\Delta_t = t-S_{i,M_i}$ and replacing $A, B,$ and $C$ by their relative quantities, we get \eqref{Lem1eq}. Note that $1/\bar{n}-l>0$, thus, $n_t$ is a non-decreasing function of the age $\Delta_t$. This completes the proof. 

When $h=0$ (i.e., $A=0$), Riccati differential equation in \eqref{rec_eqn} reduces to
\begin{align}
& \frac{dn_t}{dt}=-2\theta n_t+\sigma^2,\\
\Longleftrightarrow & \frac{d(n_t - \frac{\sigma^2}{2\theta})}{dt}=-2\theta (n_t- \frac{\sigma^2}{2\theta}).
\end{align} Using the same technique from \eqref{de_eq17} -- \eqref{remote_10} with initial condition $n_{S_{i,M_i}}=0$,
we get \eqref{Lim2_eq}. 

\section{Proof of Corollary \ref{cor-ass}}\label{cor-assapp}

The proof is to find out the function $v(\delta)$ such that Assumption \ref{ass1} holds.
Corollary (a) easily holds by taking $v(\delta)=1$. 
For Corollary \ref{cor-ass} (b), note that by H\"older's inequality, if $ \mathbb{E} \left[  Y^{n+1} \right]<\infty$, then $ \mathbb{E} \left[  Y^{i} \right]<\infty$ for all $i \le n+1$ \cite[pp. 189]{resnick2019probability}.  
We choose  
\begin{equation}
    v(\delta) = \left\{
\begin{array}{lll}
  m^{n+1} & \text{if } \delta<\bar{m} ,\\
  \delta^{n+1}  & \text{if } \delta\ge \bar{m},
\end{array}
\right. 
\end{equation} where $\bar{m}$ is a value that satisfies 
\begin{equation}
\mathbb{E} \left[ (1+ \frac{\bar{z}+\bar{x}}{\bar{m}}+\frac{Y}{\bar{m}})^{n+1}  \right] < \frac{\rho}{\alpha},
\end{equation} with $\rho$ satisfying $\alpha<\rho<1$.
In other words, as long as $v(\delta) = \Theta (\delta^{\bar{m}})$ for $\bar{m}\ge n+1$, Assumption \ref{ass1} is satisfied (by taking $m=1$). Note that $ \mathbb{E} \left[  Y^{i} \right]<\infty$ is also the sufficient and necessary condition that the function $G(\delta)$ is well-defined. 
For Corollary \ref{cor-ass} (c), note that by using mean value theorem, 
\begin{equation}
\lim_{\bar{m}\rightarrow \infty}(\bar{m}+\bar{z}+\bar{x}+\bar{y})^b-m^b = 0
\end{equation}
Thus (taking $m=1$), we choose
\begin{equation}
    v(\delta) = \left\{
\begin{array}{lll}
  e^{a \bar{m}^b} & \text{if } \delta<\bar{m} ,\\
   e^{a \delta^b} & \text{if } \delta\ge \bar{m},
\end{array}
\right. 
\end{equation} where $\bar{m}$ is a value that satisfies 
\begin{equation}
e^{a \left( (\bar{m}+\bar{z}+\bar{x}+\bar{y})^b-\bar{m}^b \right)} < \frac{\rho}{\alpha}.
\end{equation}

\section{Proof of Lemma \ref{lemma_boundedbelow}}\label{lemma_boundedbelowapp}
Let $\tilde{z} \triangleq x+z$ and define
\begin{align}
\nonumber \tilde{g}_{\gamma}(\delta,\tilde{z}) \triangleq  \mathbb{E}_Y \left[ 
 \int_{\delta}^{\delta+\tilde{z}+Y} p(t) dt - (p_{\text{opt}} + \gamma ) (\tilde{z}+Y) \right].
\end{align}
It is easy to see that $\tilde{g}_{\gamma}(\delta,\tilde{z})=g_{\gamma}(\delta,x,z)$. Note that $g_{\gamma}(\delta,x,z)$ is increasing in $\delta$. Showing that $g_{\gamma}(\delta,x,z)$ is bounded from below is equivalent to showing that $\inf_{\tilde{z} \ge 0} \tilde{g}_{\gamma}(\delta,\tilde{z})>-\infty$.

The one-sided derivatives of a function $q(w)$ at $w$ is defined as 
\begin{align} 
& \delta^+ q(w) \triangleq  \lim_{\epsilon \rightarrow 0^+} \frac{q(w+\epsilon )-q(w)}{\epsilon}.\\
 & \delta^- q(w) \triangleq  \lim_{\epsilon \rightarrow 0^+} \frac{q(w)-q(w-\epsilon)}{\epsilon}.
\end{align} 

Let us denote $k(\delta,\tilde{z},y)= \int_{\delta}^{\delta+\tilde{z}+y} p(t) dt$. Then, $k(\delta,\tilde{z},y)$ is an integration of a non-decreasing function and is thus convex in $\tilde{z}$.
Thus, $g_{\gamma}(\delta,\tilde{z})$ is also convex in $\tilde{z}$. Since $k \left( \delta, \tilde{z},y \right)$ and $\tilde{g}_{\gamma}(\delta,\tilde{z})$ are both convex in $\tilde{z}$, the one-sided derivatives of $k \left( \delta, \tilde{z},y \right)$ and $\tilde{g}_{\gamma}(\delta,\tilde{z})$ at $\tilde{z}$ exist \cite[p. 709]{bertsekas1997nonlinear}. 
Also, the function $\epsilon \rightarrow [k \left(\delta, \tilde{z}+\epsilon , y \right) - k \left(\delta, \tilde{z}, y \right)] / \epsilon$ is non-decreasing in $\epsilon\in [-\theta,0)$ or $(0,\theta]$ for some $\theta>0$ \cite[Proposition 1.1.2(i)]{butnariu2000totally}. Note that $k \left(\delta, \tilde{z} \pm \theta, Y \right)$ and one-sided derivatives of $k \left(\delta, \tilde{z}, Y \right)$ are both integrable\footnote{In this paper, although we set $p(\delta)$ as a non-decreasing real-valued function in $[0,\infty)$, we allow an exception that $p(0)=-\infty$. If $p(0)=-\infty$, we will assume that there exists a small enough $\eta>0$ such that $P(Y+X<\eta)=0$, i.e., $\delta+x\ge \eta$ with probability $1$. Therefore, $\tilde{z}>0$ when $\delta=0$ and the one-sided derivatives of $k \left( \delta, \tilde{z},Y \right)$ at $\tilde{z}$ are always integrable.}. By using Dominated Convergence Theorem \cite[Theorem 5.3.3]{resnick2019probability}, we have
\begin{align}
& \nonumber \delta^+ \tilde{g}_{\gamma}(\delta,\tilde{z}) \\ = & \lim_{\epsilon \rightarrow 0^+} \frac{1}{\epsilon}  \mathbb{E} \left[k \left(\delta,\tilde{z}+\epsilon, Y \right) - k \left(\delta,\tilde{z},Y \right) \right]  - (p_{\text{opt}}+\gamma)  \\
\nonumber  = &  \mathbb{E} \left[   \lim_{\epsilon \rightarrow 0^+}    \frac{1}{\epsilon} \left(  k \left(\delta,\tilde{z}+\epsilon,Y \right) - k \left(\delta,\tilde{z},Y \right) \right) \right] - (p_{\text{opt}}+\gamma) \\
\nonumber  = &  \mathbb{E} \left[  \lim_{\tilde{z}'\rightarrow \tilde{z}^+}  p \left(  \delta+\tilde{z}'+ Y \right)  \right] - (p_{\text{opt}}+\gamma) \\
 = & \lim_{\tilde{z}'\rightarrow \tilde{z}^+}  \mathbb{E} \left[  p \left(  \delta+\tilde{z}'+Y \right)  \right] - (p_{\text{opt}}+\gamma),  \label{derivative-lemma}  \end{align} 
Similarly, for the other direction, we get 
\begin{align} 
 \delta^- \tilde{g}_{\gamma}(\delta,\tilde{z}) & = \lim_{\tilde{z}'\rightarrow \tilde{z}^-}  \mathbb{E} \left[  p \left(  \delta+\tilde{z}'+ Y \right)  \right] - (p_{\text{opt}}+\gamma).  \label{derivative2-lemma}
\end{align}
Therefore, the solution to $\inf_{\tilde{z} \ge 0} \tilde{g}_{\gamma}(\delta,\tilde{z})$ is 
\begin{align}
\tilde{z}^* = \inf \{ \tilde{z}\ge0 : \mathbb{E} \left[  p \left(  \delta+\tilde{z}+ Y \right)  \right]  \ge p_{\text{opt}}+\gamma \}.
\end{align}
Note that $\tilde{z}^*\le \inf \{ z\ge0 : p \left( z \right)  \ge p_{\text{opt}}+\gamma \}$ and thus is finite and irrelevant to $\delta,x$. 
Also, $g_{\gamma}(\delta,x,z)\ge g_{\gamma}(0,x,z)$. Therefore, $g_{\gamma}(\delta,x,z)$ is bounded from below by a constant $-\eta$. Moreover, \begin{align}J_\pi(\delta,x) \ge \mathbb{E} \left[   \sum_{j=1}^{M} -\eta \right] = -\frac{\eta}{1-\alpha}.  
\end{align} 
 
\section{Proof of \eqref{discount}}\label{discountapp}
We have
\begin{align}
 & J_{\pi,\gamma}(\delta,x) \\  \nonumber = &  \mathbb{E} \left[   \sum_{j=1}^{M_i} g_{\gamma}(\Delta_j,X_j,Z_{j})   \ \Big{|} \  \Delta_1 = \delta, X_1 = x \right]  \\
\nonumber = & \sum_{m = 1}^{\infty}  \mathbb{E} \left[ \sum_{j=1}^{M} g_{\gamma}(\Delta_j,X_j,Z_{j})   \ \Big{|} \  \Delta_1 = \delta,  X_1 = x, M = m \right] \\ \nonumber
& \times \mathbb{P} (M = m) \\
 \nonumber
  \overset{(i)}{=} & \sum_{m = 1}^{\infty}  \mathbb{E} \left[ \sum_{j=1}^{m}  g_{\gamma}(\Delta_j,X_j,Z_{j})   \ \Big{|} \  \Delta_1 = \delta, X_1 = x \right] \\ \nonumber & \times \mathbb{P} (M = m)  \end{align}\begin{align}
  \nonumber
 = & \sum_{m = 1}^{\infty} \sum_{j=1}^{m} \mathbb{E} \left[  g_{\gamma}(\Delta_j,X_j,Z_{j})   \ \Big{|} \  \Delta_1 = \delta, X_1 = x \right] \\ \nonumber & \times \alpha^{m-1}(1-\alpha)  \\
 \nonumber
\overset{(ii)}{=} &  \sum_{j=1}^{\infty}   \sum_{m = j}^{\infty}  \mathbb{E} \left[  g_{\gamma}(\Delta_j,X_j,Z_{j})   \ \Big{|} \  \Delta_1 = \delta, X_1 = x \right]  \\  \nonumber& \times \alpha^{m-1}(1-\alpha) \\
= & \sum_{j=1}^{\infty} \alpha^{j-1} \mathbb{E} \left[  g_{\gamma}(\Delta_j,X_j,Z_{j})   \ \Big{|} \  \Delta_1 = \delta, X_1 = x \right].   
\end{align}
Note that step $(i)$ occurs because the decisions $Z_{j}$ $(j = 1,...m)$ only depend on the causal information $\delta_1, x_1, z_{1}$ ... $\delta_j,x_j$ and thus $g_{\gamma}(\Delta_j,X_j,Z_{j})$ is independent of $M$ given that $M\ge j$. 
Also, step $(ii)$ is due to Lemma \ref{lemma_boundedbelow} and the rearrangement of series in \cite[Chapter 3]{rudin1976principles}.

\section{Proof of Lemma \ref{T-welldefine}}\label{T-welldefineapp} 
Since $\delta+x+\pi(\delta,x)+y$ is Borel measurable and $u$ is lower semianalytic, $u(\delta+x+\pi(\delta,x)+y,x')$ is lower semianalytic in $(\delta,x,y,x')$ \cite[Lemma 7.30]{bertsekas2004stochastic}. Thus, $\mathbb{E}_{Y,X} \left[  u(\delta+\pi(\delta,x)+Y,X) \right]$
 is lower semianalytic in $(\delta,x)$ \cite[Proposition 7.48]{bertsekas2004stochastic}. Since the function $g_{\gamma}(\cdot)$ is Borel-measurable and thus is lower semianalytic, $T_{\pi,\gamma}(\delta,x)$ is thus lower semianalytic. Since the infimum of a lower semianalytic function is still lower semianalytic \cite[Proposition 7.47]{bertsekas2004stochastic}, $T_{\gamma}u(\delta,x)$ is lower semianalytic.

\section{Proof of Lemma \ref{prop-optimal}} \label{prop-optimalapp}

The proofs of $J_{\mu_{\text{min},\gamma} } = T_{\gamma} J_{\mu_{\text{min},\gamma} }$ and $J_{\mu_{\text{max},\gamma} } = T_\gamma J_{\mu_{\text{max},\gamma} }$ are the same. Therefore, we provide the proof of $J_{\mu_{\text{min},\gamma} } = T_\gamma J_{\mu_{\text{min},\gamma} }$. For simplicity, in this subsection, we denote $\mu = \mu_{\text{min},\gamma}$, $T=T_\gamma$, and $b = b_{\text{min},\gamma}$.

We will show that $J_{\mu}$ satisfies $T J_{\mu} = J_{\mu}$. We define the q-function $Q_{\mu}(\delta,x,z)$ as the discounted cost of starting at $\delta$, using $z$ at the first stage and then using $\mu$ for the remaining stages \cite[Section 6]{bertsekas1995dynamic1}\cite[Chapter 3]{sutton2018reinforcement}. It is easy to find that 
\begin{equation}\label{q-function}
Q_{\mu}(\delta,x,z) = g(\delta,x,z) + \alpha \mathbb{E}_{X,Y} \left[  J_{\mu}(\delta+z+Y,X) \right].
\end{equation}
Thus, it is equivalent to show that for all state $\delta,x$, $J_{\mu}(\delta,x) \le Q_{\mu}(\delta,x,z)$ for all $z$. 

From \eqref{discount}, the discounted cost is equal to the stopping cost. According to the definition of $Q_{\mu}(\delta,x,z)$ and \eqref{discount},
we can directly provide the detailed expression of $Q_{\mu}(\delta,x,z)$ with regard to the total cost like \eqref{value-function}:
\begin{lemma} \label{q}
We have 
\begin{equation} 
Q_{\mu}(\delta,x,z) =   \mathbb{E} \left[  \sum_{j=1}^{M} g(\Delta_j,X_j, Z_{j})  \ \Big{|} \ \Delta_1 = \delta, X_1 = x \right], 
\end{equation}
where 
\begin{align}
&  \Delta_{j+1} = X_j + \Delta_{j}+Z_{j}+Y_{j} \ \ \ j = 1,2...M, \\
& Z_{j} =  \left\{
\begin{array}{lll}
  z & \text{if } j=1 ,\\
  \mu(\Delta_j,X_j)  & \text{if } j =2,...M.
\end{array}
\right. \label{lemma10-84}
\end{align}
\end{lemma}
Lemma \ref{q} implies that $Q_{\mu}(\delta,x,z)$ is equal to the stopping cost such that it waits for $z$ at stage $1$ and follows the same decision as $\mu$ for the stage $2,...M$.

For simplicity, we denote 
\begin{equation}
w = z - \mu(\delta,x)
\end{equation} as the waiting time difference. Note that $w$ is a simple function of $z$ with a fixed value $\mu(\delta,x)$. 
Recall that $b$ is the threshold of $\mu$ given in \eqref{const-b}. If the state addition $\delta+x< b$, then $\mu(\delta,x)=b-\delta-x>0$, otherwise $\mu(\delta,x) = 0$. Note that if $\mu(\delta,x)=0$, $w\ge 0$ since the waiting time $z\ge 0$. Thus, there are 3 different cases based on $\delta+x$, constant $b$ and $w$:

Case (a) $\delta+x<b$ and $w \ge 0$, 

Case (b) $\delta+x<b$ and $w<0$,

Case (c) $\delta+x \ge b$ (as stated before, this implies $w \ge 0$). 

\begin{figure}[t]
\includegraphics[scale=.4]{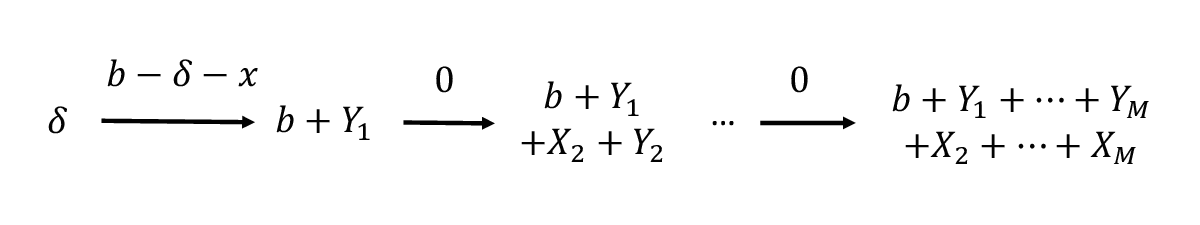}
\centering
\captionsetup{justification=justified}
\caption{Evolution diagram of $\mu$ at $i^{th}$ epoch when $\delta+x<b$.}
 \label{j-case-a}
\end{figure}

\noindent \textbf{Case (a)}: Since $\delta+x<b$, at stage $1$, $\mu(\delta,x)=b-\delta-x$, and for the remaining stages, $\mu$ chooses zero-wait. See Fig. \ref{j-case-a} for the diagram of evolution. According to \eqref{value-function}, 
\begin{align}
\nonumber J_{\mu}(\delta,x) = & \mathbb{E} \Bigg{[}   \int^{b+ Y_1+ \sum_{j=2}^{M} (X_j+Y_{j}) }_\delta  p(t) dt   \\     - &  (p_{\text{opt}} +\gamma)  \bigg ( b-\delta+Y_1+ \sum_{j=2}^{M} (X_j+Y_{j})  \bigg )  \Bigg{]},   \label{value-sum-a}
\end{align} 
where $\mathbb{E} = \mathbb{E}_{M,Y_{1},X_{2}, \ldots, X_{M},Y_{M}}$.

\begin{figure}[t]
\includegraphics[scale=.37]{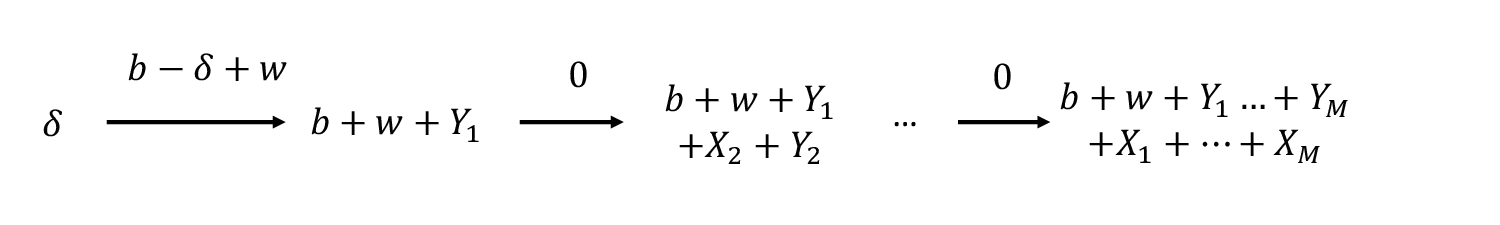}
\centering
\captionsetup{justification=justified}
\caption{Evolution diagram of the policy described in Lemma \ref{q} when $\delta +x <b$ and $w\ge 0$.}
 \label{qq-case-a}
\end{figure}
Since $w\ge 0$ in Case (a), the second age state $\Delta_2 +X_2 = b+w+Y_{1}+X_2 \ge b$, the waiting times at stage $2,3,...$  are thus $0$ (see Fig. \ref{qq-case-a} for the evolution diagram).
 This gives 
 \begin{align}
\nonumber & Q_\mu (\delta,x,z) =  \mathbb{E} \Bigg [   \int^{b+ w+Y_1+ \sum_{j=2}^{M} (X_j+Y_{j}) }_\delta  p(t) dt  \\  & -   (p_{\text{opt}} +\gamma)   \bigg ( b-\delta+w+Y_1+ \sum_{j=2}^{M} (X_j+Y_{j})  \bigg )  \Bigg ]. \label{q-mu-yz-a}
\end{align} 
It is obvious that $Q_\mu (\delta,x,\mu(\delta,x))=J_{\mu}(\delta,x)$.

\begin{figure}[t]
\includegraphics[scale=.4]{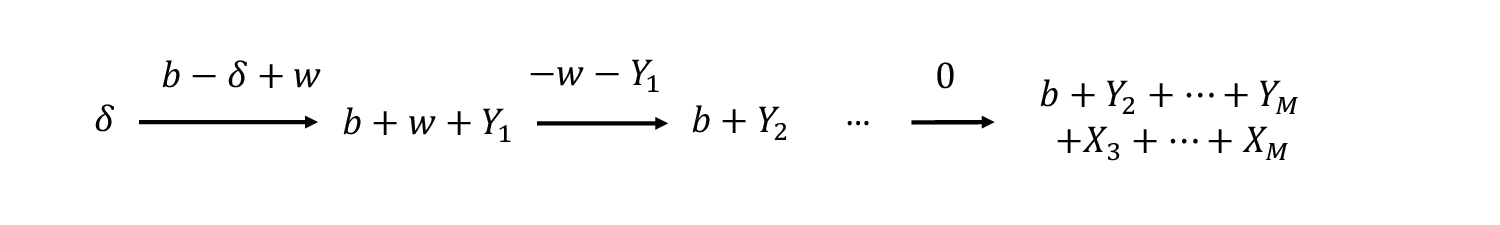}
\centering
\captionsetup{justification=justified}
\caption{Evolution diagram of the policy described in Lemma \ref{q} when $\delta<b$ and $w< 0$ given that $Y_{i,1}<-w$.}
 \label{qq-case-b}
\end{figure}
 
Note that in Case (b) and (c), $Q_{\mu}(\delta,x,z)$ does not satisfy \eqref{q-mu-yz-a}. Thus, for convenience, we rewrite
\begin{align}
\nonumber & q(w) =  \mathbb{E} \Bigg[   \int^{b+ w+Y_1+ \sum_{j=2}^{M} (X_j+Y_{j}) }_\delta  p(t) dt  \\    & -   (p_{\text{opt}} +\gamma)   \bigg( b-\delta+w+Y_1+ \sum_{j=2}^{M} (X_j+Y_{j})  \bigg)  \Bigg ].\label{disturb-delta}\end{align} 
For the ease of description, we denote $Y'$ as $Y_1+ \sum_{j=2}^{M} (X_j+Y_{j})$. 
Then, we rewrite the function that is inside the expectation of $q(w)$ in \eqref{disturb-delta}:
\begin{equation}\label{F}
f \left( w,Y' \right) =   \int^{b+ w+Y' }_\delta  p(t) dt   - (p_{\text{opt}} +\gamma)  \left( b-\delta+w+Y' \right).\end{equation}
It is obvious that $q(0)=J_{\mu}(\delta,x)$. Function $q(w)$ is the cost of the policy that waits $z=w+\mu(\delta,x)$ at stage $1$ and does not wait at stage $2,...M$ where we allow $w \le 0$ in \eqref{disturb-delta}.

By using the same technique in Appendix \ref{lemma_boundedbelowapp}, we have for all $w>-b$ (note that in case (a) and case (b), $b>0$), the one-sided derivatives of $q(w)$ satisfy
\begin{align}
 \delta^+ q(w)  
& = \lim_{x\rightarrow w^+}  \mathbb{E} \left[  p \left(  b+x+Y' \right)  \right] - (p_{\text{opt}} +\gamma), \label{derivative}  \\
 \delta^- q(w) & = \lim_{x\rightarrow w^-}  \mathbb{E} \left[  p \left(  b+x+ Y' \right)  \right] - (p_{\text{opt}} +\gamma).  \label{derivative2}
\end{align}
Since $\delta+x<b$, $b>0$. From the definition of threshold $b$ in \eqref{const-b} and $b>0$, we have
\begin{align}
& \lim_{w'\rightarrow 0^+}  \mathbb{E} \left[ p(b+w'+Y')  \right] - (p_{\text{opt}} +\gamma) \ge 0, \\
&  \lim_{w'\rightarrow 0^-}  \mathbb{E} \left[ p(b+w'+Y')  \right] - (p_{\text{opt}} +\gamma) \le 0. \label{derivative4}
\end{align}  
  By \eqref{derivative}-\eqref{derivative4}, $w=0$ is the local minimum of $q(w)$. Since $q(w)$ is convex, $w=0$ is the global optimum of $q(w)$ with $q(0) = J_{\mu}(\delta,x)$, and $q(w)$ is non-decreasing at $w\ge 0$ and non-increasing at $w\le 0$. Thus, $Q_{\mu}(\delta,x,z) = q(w) \ge J_{\mu}(\delta,x)$ for all $w \ge 0$ (i.e., $z\ge \mu(\delta)$). This proves Case (a). 

\noindent \textbf{Case (b)}: Since $\delta+x<b$ as well, $J_{\mu}$ satisfies \eqref{value-sum-a}. In Case (a) where $w \ge0 $, $Q_{\mu}(\delta,x,z) = q(w)$.  
However, in Case (b) where $w<0$, $Q_{\mu}(\delta,x,z)$ is not equal to $q(w)$; because there is a probability that $\Delta_2+X_2 = b+Y_{1}+X_2 <b$, which leads to the waiting time at stage $2$ not $0$. This difference makes our problem challenging. To show that $Q_{\mu}(\delta,x,z) \ge J_\mu(\delta,x)$ in Case (b), we will derive $Q_{\mu}(\delta,x,z)$ using the law of total expectation based on $Y_{1}+X_2$ and $M$. 

First, observe that 
\begin{align}
\nonumber & Q_{\mu}(\delta,x,z) - J_{\mu}(\delta,x) \\ = & \left( Q_{\mu}(\delta,x,z) - q(w) \right) + \left( q(w) - J_{\mu}(\delta,x) \right). \label{xyz}
\end{align}

In Case (a), we have already shown that $q(w) - J_{\mu}(\delta,x)$ is positive. Then, we focus on $Q_{\mu}(\delta,x,z) - q(w)$.


We use law of iteration on $M$ and then iterate on $Y_{1}+X_2$ given that $M >1$, 
and we have
\begin{align}
 q(w)  = & \mathbb{E} \left[ f\left( w,Y' \right)  \ \Big{|} \  M=1  \right] \mathbb{P}(M =1) \label{q-delta1}  \\ 
 \nonumber & + \mathbb{E} \left[ f\left( w,Y' \right)  \ \Big{|} \  M>1,Y_1+X_2 \ge -w  \right]  \\ & \times \mathbb{P}(M >1,Y_{1}+X_2 \ge -w)  \label{q-delta2}   \\ 
 \nonumber & +  \mathbb{E} \left[ f\left( w,Y' \right)  \ \Big{|} \  M>1,Y_{1}+X_2< -w  \right] \\ & \times \mathbb{P}(M >1, Y_{1}+X_2 < -w).  \label{q-delta3}
\end{align} Note that we denote $Y'=Y_1+\sum_{j=2}^{M} (Y_{j}+X_j)$ for the ease of descriptions.

Now, we analyse $Q_{\mu}(\delta,x,z)$. We rewrite $Q_{\mu}(\delta,x,z)$ in Lemma \ref{q}:
\begin{equation*} 
Q_{\mu}(\delta,x,z) =   \mathbb{E} \left[  \sum_{j=1}^{M} g(\Delta_j,X_j, Z_{j})  \ \Big{|}  \    \Delta_j = \delta,X_1 = x \right] ,
\end{equation*}
with waiting time $Z_{j}$ to be $\mu(\delta,x)+w$ at stage $j=1$ and $\mu(\Delta_j,X_2)$ for remaining stages.

We write $Q_{\mu}(\delta,x,z)$ according to the same iterations as parallel to \eqref{q-delta1} \eqref{q-delta2} \eqref{q-delta3}:
\begin{align}
\nonumber & Q_{\mu}(\delta,x,z) \\ = & \mathbb{E} \left[   \sum_{j=1}^{M} g(\Delta_j, X_j, Z_{j})  \ \Big{|} \  M=1  \right] \mathbb{P}(M =1) \label{qmu-delta1}  \\ \nonumber
 & + \mathbb{E} \left[  \sum_{j=1}^{M} g(\Delta_j, X_j, Z_{j})  \ \Big{|} \  M>1,Y_{1}+X_2\ge -w  \right] \\ & \times \mathbb{P}(M >1, Y_{1}+X_2 \ge -w)  \label{qmu-delta2}   \\  \nonumber
& +  \mathbb{E} \left[  \sum_{j=1}^{M} g(\Delta_j,X_j, Z_{j})  \ \Big{|} \  M>1,Y_{1}+X_2< -w  \right]  \\ & \times \mathbb{P}(M >1, Y_{1}+X_2 < -w).  \label{qmu-delta3}
\end{align}
If $M=1$, then there is no sample from stage $2$. Thus, 
\begin{align}
 & \mathbb{E} \left[  \sum_{j=1}^{M} g(\Delta_j, X_j, Z_{j})  \ \Big{|} \ \Delta_1 = \delta, X_1 = x, M = 1 \right]  \\ = & g(\delta,x,z)\\
 = &  \mathbb{E} \left[ f\left( w,Y' \right)  \ \Big{|} \  M=1  \right]. \label{case-b1}
\end{align}
This implies that the right hand side of \eqref{q-delta1} and the right hand side of \eqref{qmu-delta1} are equal. 

If $M>1$ and $Y_{1}+X_2\ge -w$, then $\Delta_2+X_2 = b+w+Y_1+X_2\ge b$. From \eqref{lemma10-84} in Lemma \ref{q}, $\Delta_2+X_2 \ge b$ implies that the waiting time at stage $2,3,...$ is $0$. Thus, Case (b2) is equivalent to Case (a), which gives
\begin{align}
& \mathbb{E} \left[  \sum_{j=1}^{M} g(\Delta_j,Z_{j})  \ \Big{|} \ \Delta_1 = \delta, M > 1, Y_{1}+X_2\ge -w \right] \\
 = &  \mathbb{E} \left[ f\left( w,Y' \right)  \ \Big{|} \  M>1, Y_{1}+X_2\ge -w  \right].  \label{case-b2}
\end{align} This implies that \eqref{q-delta2} and \eqref{qmu-delta2} are equal. 

If $M>1$ and $Y_{1}+X_2< -w$, then Lemma \ref{q} implies that the waiting time at stage $2$ is equal to
\begin{equation}\label{small-delay}
\mu(b+w+Y_{1}, X_2) = b-(b+w+Y_{1}+X_2) = -w - (Y_{1}+X_2), 
\end{equation} which is strictly larger than $0$.

Note that at stage $2$, the policy has already added up the total time to $b$. Then, the waiting time for stage $3,4...$ is $0$ because $\Delta_3\ge b$ and age is increasing. See Fig. \ref{qq-case-b} for the diagram of evolution.
Thus, given that $M>1$ and $Y_{1}+X_2< -w$, the expression of $Q_\mu(\delta,x,z)$ in \eqref{qmu-delta3} and the waiting time \eqref{small-delay} tells that
\begin{align}
 & \mathbb{E} \left[  \sum_{j=1}^{M} g(\Delta_j,X_j, Z_{j})  \ \Big{|} \  M > 1, Y_{1}+X_2 < - w \right]  \\ \nonumber
 = &  \mathbb{E} \Bigg{[}  \int^{b+ w + Y_{1}+X_2+(-w-Y_{1}-X_2)+ Y_2+\sum_{j=3}^{M} (Y_{j}+X_j) }_\delta  p(t) dt \\  \nonumber & - (p_{\text{opt}} +\gamma)   \bigg ( b-\delta+w+ Y_{1}+X_2+(-w-Y_{1}-X_2) \\ & + Y_2+\sum_{j=3}^{M} (Y_{j}+X_j)   \bigg )    \ \Big{|} \  M>1, Y_{1}+X_2 < -w  \Bigg{]}  \\ \nonumber = &  \mathbb{E} \Bigg [ \int^{b+ Y_2+\sum_{j=3}^{M} (Y_{j}+X_j) }_\delta  p(t) dt  - (p_{\text{opt}} +\gamma)   \bigg ( b-\delta+Y_2  \\ & +\sum_{j=3}^{M} (Y_{j}+X_j)  \bigg)    \ \Big{|} \  M>1, Y_{1}+X_2 < -w  \Bigg ] \\ \nonumber
 = & \mathbb{E} \Bigg [ \int^{b+Y_2+\sum_{j=3}^{M} (Y_{j}+X_j)}_\delta  p(t) dt   - (p_{\text{opt}} +\gamma)   \bigg ( b-\delta+Y_2 \\ & +\sum_{j=3}^{M} (Y_{j}+X_j)  \bigg )   \ \Big{|} \  M>1  \Bigg ] \\ \nonumber
  \overset{(i)}{=} & \mathbb{E} \Bigg[ \int^{b+ Y_2+\sum_{j=3}^{M+1} (Y_{j}+X_j) }_\delta  p(t) dt   - (p_{\text{opt}} +\gamma)   \bigg ( b-\delta+Y_2 \\ & +\sum_{j=3}^{M+1} (Y_{j}+X_j)  \bigg ) \Bigg] \end{align}\begin{align}  \nonumber 
 \overset{(ii)}{=} & \mathbb{E} \Bigg [ \int^{b+ Y_1+\sum_{j=2}^{M} (Y_{j}+X_j) }_\delta  p(t) dt   - (p_{\text{opt}} +\gamma)  \bigg( b-\delta+Y_1 \\ & +\sum_{j=2}^{M} (Y_{j}+X_j) \bigg) \Bigg] \\ \overset{(iii)}{=} & J_{\mu}(\delta,x). \label{case-b3}
\end{align}
Note that $M$ is geometric distributed and thus $M$ given that $M>1$ and $M +1$ have the same distribution. This implies $(i)$.
Condition $(ii)$ is because $Y_{j}$'s are i.i.d., and $(iii)$ is directly from the definition of $J_{\mu}(\delta,x)$ in \eqref{value-sum-a}. 

Having considered Case (b1)-(b3), to analyze $Q_{\mu}(\delta,x,z)-q(w)$, we compare \eqref{q-delta1},\eqref{q-delta2},\eqref{q-delta3} with 
\eqref{case-b1},\eqref{case-b2},\eqref{case-b3}. Note that by Case (b1) and Case (b2), \eqref{q-delta1},\eqref{q-delta2} are cancelled out by \eqref{case-b1},\eqref{case-b2}. We finally get  
\begin{align}
\nonumber & Q_{\mu}(\delta,x,z) - q(w)  \\ \nonumber
= &  J_{\mu}(\delta,x) \mathbb{P}(M >1, Y_{1}+X_2<- w)  \\ \nonumber & -  \mathbb{E} \left[ f \bigg( w,Y_1+\sum_{j=2}^{M} (Y_{j}+X_j) \bigg)  \ \Big{|} \  M>1, Y_{1}+X_2<-w  \right] \\ & \times \mathbb{P}(M >1, Y_{1}+X_2 <-w). \label{F2}
\end{align}
By the definition of $f$ in \eqref{F}, 
\begin{align}
\nonumber  \mathbb{E} & \left[  f \bigg ( w,Y_1+\sum_{j=2}^{M} (Y_{j}+X_j) \bigg )  \ \Big{|} \  M>1, Y_{1}+X_2<-w  \right]   \\ \nonumber
\triangleq & \mathbb{E} \Bigg [ \int^{b+ w +Y_1+\sum_{j=2}^{M} (Y_{j}+X_j) }_\delta  p(t) dt   - (p_{\text{opt}} +\gamma)   \bigg ( b-\delta \\ & +w +Y_1+\sum_{j=2}^{M} (Y_{j}+X_j)  \bigg )   \ \Big{|} \  M>1, Y_{1}+X_2<-w  \Bigg ] \\ \nonumber
  \overset{(i)}{=} &  \mathbb{E} \Bigg [ \int^{b+ w + Y_1+\sum_{j=2}^{M+1} (Y_{j}+X_j) }_\delta  p(t) dt   - (p_{\text{opt}} +\gamma)  \bigg ( b-\delta \\ & +w+Y_1+\sum_{j=2}^{M+1} (Y_{j}+X_j)  \bigg )   \ \Big{|} \ Y_{1}+X_2<-w  \Bigg] \\ \nonumber = &  \mathbb{E} \Bigg [ \int^{b+( w+Y_{1}+X_2)+ Y_2+ \sum_{j=3}^{M+1} (X_j+Y_{j}) }_\delta  p(t) dt \\ \nonumber &   - (p_{\text{opt}} +\gamma)   \bigg( b-\delta+( w+Y_{1}+X_2)  \\ &  + Y_2 + \sum_{j=3}^{M+1} (X_j+Y_{j} )  \bigg )  \ \Big{|} \ Y_{1}+X_2<-w  \Bigg] \label{difference} \\
  \overset{(ii)}{=} &  \mathbb{E}_{Y_{1},X_2} \left[ q(w+Y_{1}+X_2)  \ \Big{|} \ Y_{1}+X_2<-w    \right]. \label{difference2}
\end{align} 
Here $(i)$ is because $M$ has geometric distribution, and thus $M$ given that $M>1$ and $M +1$ have the same distribution. Since $Y_{j}$'s are independent, $(ii)$ is because $Y_2+\sum_{j=3}^{M+1} (Y_{j}+X_j)$ in \eqref{difference} and $Y_1+\sum_{j=2}^{M} (Y_{j}+X_j)$ inside the definition of $J_\mu(\delta) $ in \eqref{disturb-delta} have the same distributions. 
We have shown in Case (a) that $q(w)$ is decreasing at $w \le 0$. Since $Y_{1}+X_2<-w$ and $w<0$, we have
\begin{equation}
w \le w+Y_{1}+X_2< 0.
\end{equation}
Thus, 
\begin{equation}\label{q-small}
q(w+Y_{1}+X_2)\le q(w).
\end{equation}
Then, \eqref{difference2} and \eqref{q-small} gives 
\begin{align}
\nonumber & \mathbb{E} \left[ f\bigg ( w,Y_1+\sum_{j=2}^{M} (Y_{j}+X_j) \bigg )  \ \Big{|} \  M>1, Y_{1}+X_2<-w  \right]  \\  
= & \mathbb{E}_{Y_{1},X_2} \left[ q(w+Y_{1}+X_2)  \ \Big{|} \ Y_{1}+X_2<-w    \right]\\ 
\le & \mathbb{E}_{Y_{1},X_2} \left[ q(w)  \ \Big{|} \   Y_{1}+X_2<-w  \right]\\  = & q(w). \label{conclude}
\end{align} 
Thus, \eqref{F2} and \eqref{conclude} give 
\begin{align}
\nonumber  & Q_{\mu}(\delta,x,z) - q(w)  \\ 
\ge & \left( J_{\mu}(\delta,x) - q(w) \right)  \mathbb{P}(M>1, Y_{1}+X_2 <-w). \label{conclude3}
\end{align}
Note that $ q(w) - J_{\mu}(\delta,x)$ is already analyzed in Case (a) and is positive. 
Finally, \eqref{xyz}, \eqref{conclude3}, and $q(w)- J_\mu (\delta,x) \ge 0$ give 
\begin{align}
\nonumber & Q_{\mu}(\delta,x,z) - J_{\mu}(\delta,x)  \\ \ge \nonumber
& \left( J_{\mu}(\delta,x) - q(w) \right)  \mathbb{P}(M>1, Y_{1}+X_2 <-w) \\ \nonumber &
 +q(w) - J_{\mu}(\delta,x) \\ \nonumber = & \left( q(w) - J_{\mu}(\delta,x) \right)  \left( 1-   \mathbb{P}(M>1, Y_{1}+X_2 <-w) \right) \\ \ge& 0,
\end{align} 
which completes Case (b).

\noindent \textbf{Case (c)} Case (c) is similar with Case (a).

Since $\delta+x\ge b$, $\mu(\delta,x)=0$, which means that the policy $\mu$ chooses zero wait all the stages $j = 1,...M$. Thus,  
\begin{align}
\nonumber & J_{\mu}(\delta,x) = \mathbb{E} \Bigg[   \int^{\delta+ x+Y_1+\sum_{j=2}^{M} (Y_{j}+X_j) }_\delta  p(t) dt    \\ \nonumber & - (p_{\text{opt}} +\gamma)  \bigg ( x+  Y_1+\sum_{j=2}^{M} (Y_{j}+X_j)  \bigg )  \Bigg] \\ = & q(\delta+x-b).   \label{nodisturb-delta2}
\end{align} 
 
Since $\mu(\delta,x)=0$, $w\ge 0$. From Lemma \ref{q}, for all $z\ge 0$, $Q_{\mu}(\delta,x,z)$ is the cost that chooses $z$ at stage $1$ but does not wait from stages $2,3,...$. We can get

\begin{align}
\nonumber & Q_{\mu}(\delta,x,z) = \mathbb{E} \Bigg [   \int^{\delta+x+ w+ \sum_{j=1}^{M} Y_{j} }_\delta  p(t) dt \\ \nonumber &  - (p_{\text{opt}} +\gamma)    \bigg( x+w+Y_1+\sum_{j=2}^{M} (Y_{j}+X_j)  \bigg)  \Bigg ] \\  = & q(w+\delta+x-b).\label{disturb-delta2} 
\end{align} 
Since $\delta\ge b$, $w+\delta+x-b\ge 0$ for all $w\ge 0$. Using the same technique as in Appendix \ref{lemma_boundedbelowapp}, $q(w+\delta+x-b)$ is convex and non-decreasing in $w\ge 0$. This gives $q(w+\delta+x-b)\ge q(\delta+x-b)$ for all $w\ge 0$. 
 From \eqref{nodisturb-delta2} and \eqref{disturb-delta2}, we finally get $Q_{\mu}(\delta,x,z) \ge J_{\mu}(\delta,x)$ for all $z$. 

By considering Case (a)-(c), we have shown that $Q_{\mu}(\delta,x,z) \ge J_{\mu}(\delta,x)$ for all $z$, which completes the proof of $TJ_\mu = J_\mu$.
Now, we return the notation $J_\mu$ back to $J_{\mu_{\text{min},\gamma}}$, $J_{\mu_{\text{max},\gamma}}$.
By the definition of $J_\mu$ in \eqref{value-sum-a} and \eqref{derivative}-\eqref{derivative4}, it is easy to show that $J_{\mu_{\text{min},\gamma}}(\delta,x)= J_{\mu_{\text{max},\gamma}}(\delta,x)$. Moreover, similar to \eqref{value-sum-a}, for any probability $\lambda \in [0,1]$, we have $J_{\tilde{\mu}_{\lambda,\gamma}}(\delta,x) = \lambda J_{\mu_{\text{min},\gamma}}(\delta,x)+(1-\lambda)J_{\mu_{\text{max},\gamma}}(\delta,x)$. Therefore, $J_{\tilde{\mu}_{\lambda,\gamma}} =J_{\mu_{\text{min},\gamma}}=J_{\mu_{\text{max},\gamma}}$. In conclusion, we have completed the proof of Lemma \ref{prop-optimal}.

\section{Proof of Lemma \ref{preserves}}\label{preservesapp}

Since $\pi\in\Pi_i$, $J_{\pi,\gamma}(\delta,x)$ is Borel measurable \cite{bertsekas2004stochastic} and thus is lower semianalytic. It remains to show that $J_{\pi,\gamma}(\delta,x)$ is bounded by $v(\delta)$.
By Lemma \ref{lemma_boundedbelow}, $J_{\pi,\gamma} (\delta,x)\ge -\eta/(1-\alpha)$ for all $\delta,x$. Also, $v(\delta)$ is increasing and $v(0)>0$. Thus,  
\begin{equation}
 \frac{J_{\pi,\gamma}(\delta,x)}{v(\delta)} \ge - \frac{\eta}{v(0)(1-\alpha)} \ge -\infty.\label{boundabove}
\end{equation}
Then, we will show that $ J_{\pi,\gamma}(\delta,x) / v(\delta)$ is upper bounded. From Assumption \ref{ass1} (b) and $v$ is increasing, for all $n\ge 1$,  
\begin{align}
\nonumber & \mathbb{E} \left[ v(\Delta_{n+m})  \right] \\
\nonumber = & \mathbb{E} \left[ v(\Delta_{n+m-1}+X_{n+m-1}+Z_{n+m-1}+Y_{n+m-1}) \right]\\
\nonumber \le & \mathbb{E} \left[ v(\Delta_{n+m-1}+\bar{x}+\bar{z}+Y_{n+m-1})  \right]\\
\nonumber & \cdots \\
\nonumber \le & \mathbb{E} \left[ v(\Delta_{n}+m\bar{x}+m\bar{z}+Y_{n} + \cdots +Y_{n+m-1} )  \right]\\
 \overset{(i)}{\le} & \frac{\rho}{\alpha^m} \mathbb{E} \left[ v(\Delta_{n})  \right],
\end{align}   where $(i)$ is from Assumption \ref{ass1}, and that $\Delta_n$ is independent from $Y_{n}, \cdots, Y_{n+m-1}$.

Then, we look at the $n^{th}$ term in \eqref{discount}. Note that $\Delta_1 = \delta$ and the positive function $G(\delta)$ is denoted in Assumption \ref{ass1}. Also, $g_\gamma(\delta,x,z)$ is upper bounded by $G(\delta)$ plus a constant $|p_{\text{opt}}+\gamma|(\bar{x}+\bar{z}+\mathbb{E} \left[ Y \right] )\triangleq c$ that is not related to $\delta$. From Assumption \ref{ass1} (a), there exists $k>0$ such that $G(\delta)/v(\delta)\le k$. For all $n\ge1$,
\begin{align}
\nonumber & \alpha^{n-1} \mathbb{E} \left[ g(\Delta_{n},X_n,Z_{n})  \right] \\
\nonumber \le & \alpha^{n-1} \left( \mathbb{E} \left[ G(\Delta_{n})  \right] +c \right) \\
\nonumber \le & \alpha^{n-1} \cdot \left(  k \cdot  \mathbb{E} \left[ v(\Delta_{n})  \right] +c \right) \\
\nonumber \le & \alpha^{n-1} \cdot \left( k \cdot  \frac{\rho}{\alpha^m} \mathbb{E} \left[ v(\Delta_{n-m})   \right]  +c \right) \\
\nonumber & \cdots \\
\le &  k \cdot  \frac{\rho^{\lfloor \frac{n-1}{m} \rfloor +1 }}{\alpha^m} v(\Delta_1) +\alpha^{n-1}c \ \ \ \text{given that } \Delta_1 = \delta.
\end{align} 
Thus, from \eqref{discount}, 
\begin{align}
\nonumber J_\pi(\delta,x) = & \sum_{n=1}^{\infty} \alpha^{n-1} \mathbb{E} \left[  g(\Delta_n,X_n,Z_{n})   \ \Big{|} \  \Delta_1 = \delta, X_1 = x \right]  \\ \nonumber  \le & \frac{ m k v(\delta)}{\alpha^m} \sum_{n=1}^{\infty}    \rho^{n}  + \sum_{n=1}^{\infty}  \alpha^{n-1}c  \\   =  & \frac{m k \rho}{\alpha^m (1-\rho) } v(\delta) + \frac{c}{1-\alpha}. \label{boundbelow}
\end{align} Thus, $J_{\pi,\gamma}(\delta,x)/v(\delta)$ is bounded from above. By \eqref{boundabove} and \eqref{boundbelow}, we immediately get $J_{\pi,\gamma} \in B(\Lambda)$.

\section{Proof of Lemma \ref{contraction-unique}}\label{contraction-uniqueapp}

The proof of Lemma \ref{contraction-unique} is modified from \cite{bertsekas1995dynamic2}.
While \cite[Assumption 1.5.1]{bertsekas1995dynamic2} assumes countable state space and action space, we show that Lemma \ref{contraction-unique} also holds in uncountable state space and action space. 

We denote $z \triangleq \pi(\delta,x)$. By Assumption \ref{ass1}, $z\le \bar{z}$ and $x\le \bar{x}$. Note that $g_{\gamma}(\delta,x,z)\le G(\delta)+c$, where $c=|p_{\text{opt}}+\gamma|(\bar{x}+\bar{z}+\mathbb{E} \left[ Y \right] )$. 
Then, for all $u\in B(\Lambda)$,  
\begin{align}
\nonumber & T_{\pi,\gamma} u (\delta,x) \\ \nonumber = & g_\gamma(\delta,x,z)+ \alpha  \mathbb{E} \left[ u(\delta+x+z+Y,X)  \right]\\
\nonumber \le & G(\delta)+ \alpha  \mathbb{E} \left[  \left| \frac{u(\delta+x+z+Y,X)}{v(\delta+x+z+Y)} \right| v(\delta+x+z+Y)  \right] +c \\
\nonumber \le & G(\delta)+ \alpha \| u \|  \mathbb{E} \left[  v(\delta+x+z+Y) \right] +c \\
 \overset{(i)}{\le} & G(\delta)+ \frac{\rho}{\alpha^{m-1}} \| u \| v(\delta)+c, \label{tpiu-upper}
\end{align} where $(i)$ is from Assumption \ref{ass1} and that $v(\delta)$ is increasing. 

On the other hand, since $g_{\gamma}(\delta,x,z)\ge -\lambda$, we have 
\begin{align}
& \nonumber T_{\pi,\gamma} u (\delta,x) \\ \nonumber \ge & -\eta - \alpha  \mathbb{E} \left[ \left| \frac{u(\delta+x+z+Y,X)}{v(\delta+x+z+Y)} \right| v(\delta+x+z+Y)  \right] \\
\ge & -\eta - \frac{\rho}{\alpha^{m-1}} \| u \| v(\delta). \label{tpiu-lower}
\end{align} 
Thus, divide $v(\delta)$ and take maximum over $\delta$ in \eqref{tpiu-upper} and \eqref{tpiu-lower}, we finally get 
\begin{align}
\| T_{\pi,\gamma}  u \| & \le \max\{ \| G \|+\frac{c}{v(0)}, \frac{\eta}{v(0)} \}  + \frac{\rho}{\alpha^{m-1}} \| u \|  < \infty. 
\end{align} Thus, $T_{\pi,\gamma}  u \in B(\Lambda)$. 

We then show that $T_\gamma u \in B(\Lambda)$:  
\begin{align}
\nonumber & T_\gamma u (\delta,x) \\ \nonumber = & \inf_{z\in [0,\bar{z}]} g_\gamma(\delta,x,z)+ \alpha  \mathbb{E} \left[ u(\delta+x+z+Y,X)  \right]\\
\nonumber \le & \sup_{z\in [0,\bar{z}]} g(\delta,x,z) \\ \nonumber & +  \sup_{z\in [0,\bar{z}]} \alpha  \mathbb{E} \left[ \left| \frac{u(\delta+x+z+Y,X)}{v(\delta+x+z+Y)} \right| v(\delta+x+z+Y)  \right] \\ \nonumber
 \le & G(\delta)+ \alpha \| u \| \sup_{z\in [0,\bar{z}]} \mathbb{E} \left[  v(\delta+x+z+Y) \right] +c \\
 \le & G(\delta)+ \frac{\rho}{\alpha^{m-1}} \| u \| v(\delta) +c. 
\end{align} Similarly, 
\begin{align}
T_\gamma u (\delta) \ge -\eta - \frac{\rho}{\alpha^{m-1}} \| u \| v(\delta). 
\end{align} Thus, 
\begin{align}
\| T_\gamma u \| & \le \max\{ \| G \| + \frac{c}{v(0)}, \frac{\eta}{v(0)} \}  + \frac{\rho}{\alpha^{m-1}} \| u \|  < \infty, 
\end{align} which shows that $T_\gamma u\in B(\Lambda)$.

For all $u(\delta,x),u'(\delta,x)\in B(\Lambda)$ and any deterministic stationary policy $\pi_0,\cdots \pi_{m-1} \in \Pi_i$, along with Assumption \ref{ass1}, 
\begin{align}
\nonumber & \| T_{\pi_0,\gamma} \cdots T_{\pi_{m-1},\gamma} u -  T_{\pi_0,\gamma} \cdots T_{\pi_{m-1}}  u' \| \\
\nonumber \le & \sup_{\delta} \frac{ \left| \mathbb{E} \left[  u(\Delta_{m+1},X_{m+1}) -u'(\Delta_{m+1},X_{m+1})    \right] \right| }{v(\delta)} \alpha^m \\
\nonumber \le & \sup_\delta \frac{ \mathbb{E} \left[ v(\delta+m\bar{x}+m\bar{z}+\sum_{j=1}^{m}Y_j)  \right]}{v(\delta)}\cdot \alpha^m \| u-u' \| \\
\le & \rho \| u-u' \|, \label{contraction-pi}
\end{align} where the last inequality is from Assumption \ref{ass1}. This implies the contraction mapping property of $T_{\pi_0,\gamma} \cdots T_{\pi_{m-1},\gamma}$. By \eqref{contraction-pi}, we have 
\begin{align} 
& \frac{T_{\pi_0,\gamma} \cdots T_{\pi_{m-1},\gamma} u(\delta,x)}{v(\delta)} \\ \le  &\frac{T_{\pi_0,\gamma} \cdots T_{\pi_{m-1},\gamma} u'(\delta,x)}{v(\delta)} + \rho \| u-u' \|. \label{pi-to-t}
\end{align}
Taking the minimum for left and right sides of \eqref{pi-to-t} of $\pi_0,\cdots \pi_{m-1}$, respectively, we have 
\begin{align} 
& \frac{T_\gamma^m u(\delta,x)}{v(\delta)} \le \frac{T_\gamma^m u'(\delta,x)}{v(\delta)} + \rho \| u-u' \|. \label{pi-to-t2}
\end{align} Reversing $u$ and $u'$ in \eqref{pi-to-t2}, we finally get $\| T_\gamma^m u - T_\gamma^m u'\|\le \rho \|u-u'\|$, i.e., the Bellman operator $T_\gamma$ has an $m$-stage contraction mapping property with modulus $\rho<1$.
Combined with $B(\Lambda)$ being complete, the uniqueness of $T_\gamma u=u$ is shown directly by \cite[Proposition 1.5.4]{bertsekas1995dynamic2}. Thus, we complete the proof of Lemma \ref{contraction-unique}.

\section{Proof of Theorem \ref{theorem3}}\label{zerodualityapp}
According to \cite[Proposition 6.2.5]{bertsekas2003convex}, the policy $\pi^*$ along with the dual variable $\gamma^*$ is the optimal solution to \eqref{hc} if the following conditions hold:
\begin{align} 
&  \lim_{n\rightarrow \infty} \frac{1}{n}  \sum_{i=1}^{n}    \mathbb{E} \left[ D_{i,M_i} - D_{i-1,M_{i-1}}  \right] - \frac{1}{f_{\text{max}}(1-\alpha)} \ge 0, \label{duality1} \\
& \pi^* \in \Pi, \gamma^* \ge 0, \label{duality2} \\
& L(\pi^*; \gamma^*) = \inf_{\pi\in\Pi} L(\pi; \gamma^*), \label{duality3} \\
& \gamma^* \left\{ \lim_{n\rightarrow \infty} \frac{1}{n}  \sum_{i=1}^{n}    \mathbb{E} \left[ D_{i,M_i} - D_{i-1,M_{i-1}}  \right] - \frac{1}{f_{\text{max}}(1-\alpha)}  \right\} = 0.\label{duality4}
\end{align}
According to Lemma \ref{prime-lemma}, the optimal solution to the primal problem \eqref{duality3} for a given $\gamma$ is given by $l(\gamma)$. Therefore, we will seek $\gamma^*$ and $\pi^*\in l(\gamma^*)$ that satisfies \eqref{duality1}, \eqref{duality2} and \eqref{duality4}. 
According to Lemma \ref{prime-lemma}, for any optimal policy with waiting times $Z_{i,j}$'s, we have 
\begin{align}
& D_{i,M_i} - D_{i-1,M_{i-1}} =  X_{i,1} +Z_{i,1}+Y', \label{dual-app160}\\ 
& Z_{i,1}\ge \mu_{\text{min},\gamma}(Y_{i-1,M_{i-1}},X_{i,1}), \label{dual-app161} \\
 & Z_{i,1}\le  \mu_{\text{max},\gamma}(Y_{i-1,M_{i-1}},X_{i,1}),\label{dual-app162}
\end{align}
This motivates us to consider the following two cases.

\noindent \textbf{Case 1} If \eqref{rate-useless} holds,
then we take $\gamma^* = 0$ and $\pi^* = \mu_{\text{min},0}$. Then, \eqref{duality1}-\eqref{duality4} are satisfied. 

\noindent \textbf{Case 2} If the condition \eqref{rate-useless} does not hold, we will seek $\gamma^*>0$ and $\pi^*\in l(\gamma^*)$ such that  
\begin{align}
\lim_{n\rightarrow \infty} \frac{1}{n}  \sum_{i=1}^{n}    \mathbb{E} \left[ D_{i,M_i} - D_{i-1,M_{i-1}}  \right] - \frac{1}{f_{\text{max}}(1-\alpha)} =0. \label{rate-equal}
\end{align}
By \eqref{dual-app160}-\eqref{dual-app162}, we need to seek $\gamma^*>0$ such that 
\begin{align} 
\nonumber & \lim_{n\rightarrow \infty} \frac{1}{n}  \sum_{i=1}^{n}  \mathbb{E} \left[ X_{i,1}+ \mu_{\text{min},\gamma^*}(Y_{i-1,M_{i-1}},X_{i,1})+Y'   \right] \\ \nonumber \le & \frac{1}{f_{\text{max}}(1-\alpha)}\\  \le & \lim_{n\rightarrow \infty} \frac{1}{n}  \sum_{i=1}^{n}  \mathbb{E} \left[ X_{i,1}+ \mu_{\text{max},\gamma^*}(Y_{i-1,M_{i-1}},X_{i,1})+Y'   \right],  \label{sum-solution} 
\end{align}
Since the $Y_{i,j}$'s and the $X_{i,j}$'s are i.i.d., and $Y_{i-1,M_{i-1}}$ is independent of the sampling times at epoch $0,1,2,...,i-1$, \eqref{sum-solution} is equivalent to 
\begin{align}
& \nonumber \mathbb{E} \left[ X_{i,1}+\mu_{\text{min},\gamma^*}(Y_{i-1,M_{i-1}},X_{i,1})+Y'   \right]  \le  \frac{1}{f_{\text{max}}(1-\alpha)}\\  \le &   \mathbb{E} \left[ X_{i,1}+\mu_{\text{max},\gamma^*}(Y_{i-1,M_{i-1}},X_{i,1})+Y'   \right]. \label{waiting-fmax}
\end{align}
It is easy to see that $\mu_{\text{min},\gamma}(\delta,x)$ and $\mu_{\text{min},\gamma}(\delta,x)$ are non-decreasing in $\gamma$. It holds that for all $\gamma_0 \ge 0$,
\begin{align} 
\nonumber & \lim_{\gamma \rightarrow \gamma_0-} \mu_{\text{max},\gamma}(\delta,x) = \mu_{\text{min},\gamma_0}(\delta,x) \le \mu_{\text{max},\gamma_0}(\delta,x) \\ = & \lim_{\gamma \rightarrow \gamma_0 +} \mu_{\text{min},\gamma}(\delta,x).
\end{align}
By Monotone Convergence Theorem \cite[Theorem 5.3.1]{resnick2019probability}, we have 
\begin{align} 
\nonumber & \lim_{\gamma \rightarrow \gamma_0-} \mathbb{E} \left[ \mu_{\text{max},\gamma}(Y_{i-1,M_{i-1}},X_{i,1})   \right] \\ \nonumber = & \mathbb{E} \left[ \mu_{\text{min},\gamma_0}(Y_{i-1,M_{i-1}},X_{i,1})  \right] \\ \nonumber \le & \mathbb{E} \left[ \mu_{\text{max},\gamma_0}(Y_{i-1,M_{i-1}},X_{i,1})  \right] \\  = & \lim_{\gamma \rightarrow \gamma_0 +} \mathbb{E} \left[ \mu_{\text{min},\gamma}(Y_{i-1,M_{i-1}},X_{i,1})   \right]. \label{gamma-exists}
\end{align}
According to \eqref{gamma-exists}, there exists $\gamma^*$ that satisfies \eqref{waiting-fmax}. The optimal policy $\pi^*=\tilde{\mu}_{\lambda,\gamma^*}$, where $\lambda$ is defined in \eqref{theorem3-prob} to achieve \eqref{waiting-fmax}.
In both cases, the $\pi^*$ and $\gamma^*$ selected satisfy \eqref{duality1}-\eqref{duality4}.

\section{Proof of Lemma \ref{lemma-hbeta}}\label{lemma-hbetaapp}
(a) Note that $p_{\text{opt}} \in [\underline{p},\bar{p}) \cap \mathbb{R}$, so we consider $\beta \in [\underline{p},\bar{p}) \cap \mathbb{R}$.
We define 
\begin{align}
& \nonumber L_i (\pi;\beta) \\ = & \mathbb{E} \Bigg{[}  \int^{D_{i,M_i}}_{D_{i-1,M_{i-1}}}  p(\Delta_t) dt  -  \beta \left(  D_{i,M_i} - D_{i-1,M_{i-1}}  \right)   \Bigg{]}.
\end{align}
As is shown in Section \ref{per-epoc}, 
\begin{align}
f(\beta) = \inf_{\pi \in \Pi_i} L_i(\pi;\beta).
\end{align}
It is easy to see that $L_i(\beta,\pi)$ is linear in $\beta$ and thus concave. Then, $f(\beta)$ is an infimum of a sequence of concave functions and is thus concave in $\beta$. For any policy $\pi\in\Pi_i$ and $\beta_2>\beta_1$, we have $L_i(\beta_2,\pi) = L_i(\beta_1,\pi)+(\beta_2-\beta_1) \mathbb{E}[D_{i,M_i} - D_{i-1,M_{i-1}}]$. Since $\mathbb{E}[D_{i,M_i} - D_{i-1,M_{i-1}}]$ is lower bounded by a positive value, taking infimum on both sides, $f(\beta)$ is strictly decreasing. 
Therefore, $f(\beta)$ is concave and strictly decreasing in $\beta \in [\underline{p},\bar{p}) \cap \mathbb{R}$. 

(b) By Lemma \ref{dinklebach-lemma}, we have $h(p_{\text{opt}})=0$. By Lemma \ref{lemma3}, $h(p_{\text{opt}})=0$ is equivalent to $f(p_{\text{opt}})=0$. Since $f(\beta)$ is strictly decreasing, there exists a unique root to $f(\beta)=0$.

\section{Proof of Corollary \ref{equiv-zerowait}}\label{equiv-zerowaitapp}
We use w.p. $1$ as the abbreviation of ``with probability $1$''.
We first prove the backward direction when \eqref{cor2} is satisfied. For any waiting time $z\ge 0$, we have 
\begin{align}
& \mathbb{E}_{Y'} \left[  p( Y+X+z + Y') \mid Y,X \right] \label{equiv-zerowait-1} \\
\ge & \mathbb{E}_{Y'} \left[  p( Y+X + Y') \mid Y,X \right] \\
\ge & \text{ess} \inf \mathbb{E}_{Y'} \left[  p( Y+X + Y') \mid Y,X \right]  \ \ \ \text{w.p. $1$},  \\
\ge &  \frac{ \mathbb{E} \left[  \int^{Y+X + Y'}_{Y}  p(t) dt \right]  }  { \mathbb{E} \left[ X+Y' \right] }  \ \ \ \text{w.p. $1$}.
\end{align} Note that $\mathbb{E} \left[  \int^{Y+X + Y'}_{Y}  p(t) dt \right] / \mathbb{E} \left[ X+Y' \right] $ is the average age penalty for the zero-wait policy and is no smaller than that of the optimal policy. Thus, 
\begin{equation}
 \frac{ \mathbb{E} \left[  \int^{Y+X + Y'}_{Y}  p(t) dt \right]  }  { \mathbb{E} \left[ X+Y' \right] } \ge \beta. \label{equiv-zerowait-2}
\end{equation}
Therefore, \eqref{equiv-zerowait-1}-\eqref{equiv-zerowait-2} gives that w.p. $1$, for any waiting time $z$, $\mathbb{E}_{Y'} \left[  p( Y+X+z + Y') \mid Y,X \right]\ge \beta$. 
Combining with \eqref{thm1-beta} in Theorem \ref{theorem1}, the zero-wait policy is optimal. 

Then, we prove the forward direction. Suppose that by Theorem \ref{theorem1}, the zero-wait policy is optimal. Then, the zero-wait policy $Z_{i,j}(\beta)=0$ satisfies \eqref{thm1-beta} and \eqref{root} in Theorem \ref{theorem1}. Then, we have 
\begin{equation}
 \mathbb{E}_{Y'} \left[  p( Y+X + Y') \mid Y,X \right] \ge \beta.  
\end{equation} 
Thus, we have 
\begin{equation}
 \text{ess} \inf \mathbb{E}_{Y'} \left[  p( Y+X + Y') \mid Y,X \right] \ge \beta.
\end{equation}
Since zero-wait policy is optimal, in Theorem \ref{theorem1}, we have
\begin{equation}
  \beta =  \frac{ \mathbb{E} \left[  \int^{Y+X + Y'}_{Y}  p(t) dt \right]  }  { \mathbb{E} \left[ X+Y' \right] }.
\end{equation} Thus, the forward direction is shown. Overall, the proof is completed.

\section{Proof of Corollary \ref{deter}}\label{deterapp}

Suppose that $X_{i,j}=x$ and $Y_{i,j} = y$. Since the $M_i$'s are i.i.d and geometrically distributed, we use $M$ to replace $M_i$. From Corollary \ref{equiv-zerowait}, it is sufficient to show that 
\begin{equation}\label{cor-deter-s}
 \mathbb{E} \left[ p \left(y+M(x+y) \right)  \right] \ge \frac{ \mathbb{E} \left[  \int^{y+M(x+y)}_{y}  p(t) dt \right]  }  { \mathbb{E} \left[ M (x+y)  \right] }.
\end{equation}
Here, 
\begin{align}
 \mathbb{E} \left[ M (x+y) \right] &=  \frac{x+y}{1-\alpha}, \label{cor-deter-1} \\ \mathbb{E}  \left[ p  \left(y+M(x+y) \right) \right] & =  \sum_{m=1}^{\infty} p \left(y+m(x+y) \right) \alpha^{m-1} (1-\alpha). \label{cor-deter-2} \end{align} In addition,
 \begin{align}
& \nonumber \mathbb{E} \left[  \int^{y+M(x+y)}_{y}  p(t) dt \right]  \\ = &  \sum_{m=1}^{\infty}  \int^{y+m(x+y)}_{y}  p(t) dt \cdot \alpha^{m-1} (1-\alpha) \\
 = &  \sum_{m=1}^{\infty}   \sum_{n=1}^{m} \int^{y+n(x+y)}_{y+(n-1)(x+y)}  p(t) dt \cdot \alpha^{m-1} (1-\alpha) \end{align} \begin{align}
 = &  \sum_{n=1}^{\infty}   \sum_{m=n}^{\infty} \int^{y+n(x+y)}_{y+(n-1)(x+y)}  p(t) dt \cdot \alpha^{m-1} (1-\alpha)  \\
  = &  \sum_{n=1}^{\infty}  \int^{y+n(x+y)}_{y+(n-1)(x+y)}  p(t) dt \cdot \alpha^{n-1} \\
  \le &  \sum_{n=1}^{\infty} (x+y) p(y+n(x+y))  \alpha^{n-1}\\
  \overset{(i)}{=} & \mathbb{E} \left[ M (x+y)  \right]   \mathbb{E} \left[ p \left(y+ M(x+y) \right)  \right],
\end{align}
where $(i)$ comes directly from \eqref{cor-deter-1} and \eqref{cor-deter-2}. 
Thus, we have shown \eqref{cor-deter-s}. This completes our proof.

\else
\fi

\end{document}